\documentclass[12pt]{article}
\usepackage[utf8]{inputenc}
\usepackage{amsmath}
\usepackage{amsthm}											
\usepackage{amsfonts}
\usepackage{amssymb}
\usepackage{braket}
\usepackage{caption}
\usepackage{sidecap}
\usepackage{subcaption}
\usepackage{comment}
\usepackage{esint}
\usepackage[a4paper, total={6.5in, 10in}]{geometry}
\usepackage{tikz}
\usetikzlibrary{matrix}
\usepackage{graphicx}
\usepackage{mathtools}
\usepackage{pdflscape}
\usepackage{xcolor}
\graphicspath{ {images/} }
\usepackage[section]{placeins}

\DeclareMathOperator{\Op}{Op}
\DeclareMathOperator{\Ran}{Ran}
\DeclareMathOperator{\Tr}{Tr}
\DeclareMathOperator{\tr}{tr}
\DeclareMathOperator{\sgn}{sgn}
\DeclareMathOperator{\supp}{supp}

\newcommand{\norm}[1]{\left\lVert#1\right\rVert}
\newcommand{\vertiii}[1]{\left\lvert\kern-0.25ex\left\lvert\kern-0.25ex\left\lvert #1 \right\rvert\kern-0.25ex\right\rvert\kern-0.25ex\right\rvert}
\newtheorem{proposition}{Proposition}
\newtheorem{lemma}[proposition]{Lemma}

\newtheorem{theorem}[proposition]{Theorem}

\numberwithin{proposition}{section}

\newcommand{\codim}{\text{codim}}

\newcommand{\eps}{\varepsilon}

\newcommand{\aver}[1]{\langle #1 \rangle}

\newcommand{\ad}{\text{ad}}
\newcommand{\fs}{\mathfrak{S}}

\title{Asymmetric transport for magnetic Dirac equations}
\author{Solomon Quinn and Guillaume Bal}

\begin{document}

\maketitle

\begin{abstract}
    This paper 
    concerns the asymmetric transport associated with
    a low-energy interface Dirac model of graphene-type materials subject to external magnetic and electric fields. We show that the relevant physical observable, an interface conductivity, is quantized and robust to 
    a large class of perturbations. 
    These include defects that decay along or away from the interface,
    and sufficiently small or localized changes in the 
    external fields.
    An explicit formula for the interface conductivity is given by a spectral flow.
\end{abstract}

\section{Introduction}\label{sec:intro}
Topological insulators are formed by joining
two insulators together at an interface. When the two insulators are topologically distinct, the resulting system exhibits ``edge modes'' that are localized in the vicinity of the interface and propagate along it.
The transport (e.g. electronic or photonic) associated with these modes is asymmetric, meaning there is a nonzero net flow of signal (for example an electric current). 
Remarkably this asymmetric transport is robust in the presence of defects. 
For examples and applications of such phenomena, see \cite{BH, delplace, sato, Volovik, Witten}.

The first experimental evidence of topologically protected transport was the discovery of the integer quantum Hall effect \cite{klitzing1980new}. This is the phenomenon that the transverse resistivity of a two-dimensional electron gas (at low temperature and subject to a strong perpendicular magnetic field) is quantized (instead of growing linearly in the strength of the magnetic field, as predicted by the classical theory).
More recently, other materials such as graphene \cite{katsnelson2007graphene,neto2009electronic} and twisted bilayer graphene \cite{andjelkovic2018dc,carr2018relaxation,rickhaus2018transport,san2013helical} were shown to admit robust edge modes. 


In this paper, we derive a quantized conductivity
at the interface of two (distinct) sheets of graphene (or other similar materials).
In the absence of external electric and magnetic fields, the edge state dynamics 
is governed
by a two-dimensional Schr\"odinger Hamiltonian 
whose coefficients obey the appropriate honeycomb symmetry away from the interface.
It was shown in \cite{1,drouot2019characterization,FLW-ES-2015} that at low energies, 
this Hamiltonian
is well approximated by the following Dirac operator, 
\begin{align*}
    H_D = D_x \sigma_1 + D_y \sigma_2 + m \sigma_3; \qquad D_\alpha := -i \partial_\alpha, \quad m \in \fs (m_-, m_+), \quad 0 \ne m_\pm \in \mathbb{R}.
\end{align*}
Here,
$\fs (a,b) := \cup_{\alpha < \beta} \fs(a,b;\alpha,\beta)$ with
$\fs (a,b;\alpha,\beta)$ the set of smooth functions $f: \mathbb{R} \rightarrow \mathbb{R}$ 
such that $f(\lambda) = a$ (resp. $f(\lambda) = b$) whenever
$\lambda \le \alpha$ (resp. $\lambda \ge \beta$), and the Pauli matrices are given by
\begin{align*}
    \sigma_1 = \begin{pmatrix}0 &1\\1& 0 \end{pmatrix}, \qquad \sigma_2 = \begin{pmatrix}0 &-i\\i& 0 \end{pmatrix}, \qquad \sigma_3 = \begin{pmatrix}1 &0\\0& -1 \end{pmatrix}.
\end{align*}
Observe that the interface material described by $H_D$ is
made up of two insulators (given by $H_\pm := D_x \sigma_1 + D_y \sigma_2 + m_\pm \sigma_3$) that are glued together 
along a one-dimensional interface, say $\{m(x) = 0\}$.
The \emph{mass term} $m$ models the transition from one insulator to the other.
Although $H_\pm$ has a spectral gap in the interval $(-|m_\pm|, |m_\pm|)$, 
the interface material is a conductor whenever $m_+$ and $m_-$ have opposite sign.
In this case $H_D$ no longer has a spectral gap, and the energies in the interval $(-m_0, m_0)$ with $m_0 := \min\{|m_-|, |m_+|\}$ correspond to propagating edge modes as described above.
The asymmetric transport associated with this model is quantified by the following interface conductivity,
\begin{align*}
    \sigma_I (H_D) := \Tr i [H_D,P] \varphi ' (H_D).
\end{align*}
Here, $P(y)=P \in \fs (0,1)$ and $\varphi \in \fs(0,1;-m_0,m_0)$. 
Thus $\varphi '$ is a density of states whose support is contained in the interval $(-m_0,m_0)$.
See \cite{2, 3, BH, Drouot, Elbau, elgart2005equality,QB, SB-2000} for physical derivations of $\sigma_I$ as a conductivity, and for rigorous results on its stability in various settings. Explicit formulas for $\sigma_I$ (that apply to $H_D$ among other models) are derived in \cite{3,4,bal2022mathematical,QB}.
In \cite{2} 
it was shown that $2\pi\sigma_I (H_D) = \text{SF} (H_D; \alpha)= \frac{1}{2}\sgn (m_- - m_+)$, 
where $\text{SF} (H_D; \alpha)$ is the \emph{spectral flow of $H_D$ through $\alpha$} (more on this below) and $\alpha \in (-m_0, m_0)$.
Observe that the 
quantity $\sgn (m_- - m_+)$ is independent of $\alpha, P$ and $\varphi$, and is robust with respect to changes in $m$.



The goal of this paper is to extend the above
results for $H_D$ to Dirac operators with electric and magnetic fields; we will explicitly calculate the interface conductivity by relating it to a spectral flow, and prove its stability in the presence of perturbations. The electromagnetic Dirac operators are given by
\begin{align}\label{eq:md}
    H = D_x \sigma_1 + (D_y - A_2(x)) \sigma_2 + m(x) \sigma_3 + V(x) \sigma_0,
\end{align}
where $A_2 (x) = x B (x)$ 
and
\begin{align}\label{eq:dw}
    B \in \fs (B_-, B_+), \qquad m \in \fs (m_-, m_+), \qquad V \in \fs (V_-, V_+)
\end{align}
for some constants $B_\pm, m_\pm, V_\pm \in \mathbb{R}$ with $B_\pm \ne 0$.
Here, 
$\sigma_0$ is the $2 \times 2$ identity matrix.
The vector-valued function $A = (0, A_2)$ is the magnetic potential (in the Landau gauge), with $(B+xB') \hat{z} = \nabla \times A$ the magnetic field. The function $V$ represents the electric potential. See \cite[Section 4.2]{thaller2013dirac} for a derivation of these models.

We see that the Hamiltonian \eqref{eq:md} implements three \emph{domain walls} given by \eqref{eq:dw}. 
The magnetic domain wall $B$ gives rise to the unbounded term $A_2 (x)$, while the functions $m$ and $V$ are necessarily bounded.
We have already discussed the domain wall in $m$ as a mechanism for transport. To motivate the domain wall in $B$, we provide an illustrative example.

Consider a two-dimensional electron gas confined to the right-half plane and subject to a constant (and strong) orthogonal magnetic field. Away from the edge $\{x = 0\}$ each electron will move in a (small) circular path, 
meaning the gas is insulating in its bulk.
However the presence of the edge causes nearby electrons to propagate, giving rise to a current along the boundary.
The same can be said of a (no longer confined) two-dimensional electron gas subject to an orthogonal magnetic field that changes signs across the $y-$axis.

The existing literature has analyzed separately the roles of a mass term \cite{2,3} and magnetic field \cite{cornean2022bulk} in generating asymmetric transport for Dirac models.
To our knowledge, this paper is the first to combine the two effects.
We also include
the domain wall in $V$, which gives rise to an electric field $E = -\nabla V$ perpendicular to the interface that vanishes whenever $|x|$ is sufficiently large.




The spectral properties of $H$ are acquired from the following \emph{bulk Hamiltonians},
\begin{align*}
    H_\pm = D_x \sigma_1 + (D_y - xB_\pm) \sigma_2 + m_\pm \sigma_3 + V_\pm \sigma_0,
\end{align*}
which model the two insulating materials that are glued together to form the conductor with Hamiltonain $H$.
Since $H$ and $H_\pm$ are translation-invariant in $y$, we can define their Fourier transforms by
\begin{align*}
    \hat{H} (\zeta) &= D_x \sigma_1 + (\zeta - A_2(x)) \sigma_2 + m(x) \sigma_3 + V(x) \sigma_0,\\ 
    \hat{H}_\pm (\zeta) &= D_x \sigma_1 + (\zeta - xB_\pm) \sigma_2 + m_\pm \sigma_3 + V_\pm \sigma_0 
\end{align*}
for $\zeta \in \mathbb{R}.$
It is known, see e.g. \cite{4,Bony,thaller2013dirac}, that 
each operator $\mathcal{O} \in \{H, H_+, H_-, \hat{H} (\zeta), \hat{H}_+ (\zeta), \hat{H}_- (\zeta)\}$ defined above is self-adjoint with domain of definition $\mathcal{D} (\mathcal{O}) = (i-\mathcal{O})^{-1} \mathcal{H}$. Here, $\mathcal{H} := L^2 (\mathbb{R}^d) \otimes \mathbb{C}^2$, where $d = 2$ for $H$ and $H_\pm$, and $d=1$ for $\hat{H}$ and $\hat{H}_\pm$. 
Note that for a similar edge model, it was shown that the domain of definition of the Hamiltonian depends on the strength of the magnetic field \cite{cornean2022bulk}. This suggests that $\mathcal{D} (\mathcal{O})$ likely depends on $B_\pm$, though we do not investigate this issue here.

Throughout this paper,
we denote the spectrum of an operator $\mathcal{O}$ by $\sigma (\mathcal{O})$ and the resolvent set of $\mathcal{O}$ by $\rho (\mathcal{O})$.
Suppose temporarily that $V_\pm = 0$, as these constants contribute merely a uniform shift to the spectrum of $H_\pm$. Then $\hat{H}^2_\pm (\zeta) = D^2_x + (\zeta - x B_\pm)^2 + m_\pm^2 - B_\pm \sigma_3$ is block diagonal, meaning that up to shifts by $m^2_\pm - B_\pm$ and $m^2_\pm + B_\pm$, the spectra of $\hat{H}^2_\pm (\zeta)$ and $L_\pm (\zeta) := D^2_x + (\zeta - x B_\pm)^2$ are the same. But $L_\pm (\zeta)$ is the Hamiltonian for the quantum harmonic oscillator (up to rescaling) and has spectrum consisting entirely of eigenvalues and given by $\sigma (L_\pm (\zeta)) = \{(2k+1)B:k\in\mathbb{N}\}$.
The elements of $\sigma (L_\pm (\zeta))$ are known as ``Landau levels.''
Since $\sigma (L_\pm (\zeta))$ (and hence $\sigma (\hat{H}_\pm (\zeta))$) is independent of $\zeta$, 
it follows that $\sigma (H_\pm)$ is made up of a countable collection of points (all corresponding to essential spectrum now) going to infinity
in absolute value; see Lemma \ref{lemma:spectrumhat}.
The domain wall in $B$ can cause spectral gaps between Landau levels to close, with eigenvalues of $\hat{H}^2 (\zeta)$ potentially going to infinity as $\zeta \rightarrow \pm\infty$; see Lemma \ref{lemma:muInfty}.

When spectral branches of $\hat{H} (\zeta)$ do not go to infinity, they converge to elements of $\sigma (H_\pm)$ as shown by Lemma \ref{lemma:mulim}.
Each branch converges both as $\zeta \rightarrow \infty$ and $\zeta \rightarrow -\infty$ 
in the case that $B_+$ and $B_-$ have the same sign. When $B_- < 0 < B_+$ (resp. $B_+ < 0 < B_-$), the branches converge only as $\zeta \rightarrow \infty$ (resp. $\zeta \rightarrow -\infty$).
Hence no matter the signs of $B_+$ and $B_-$, the branches of spectrum do not go to infinity as $|\zeta| \rightarrow \infty$.
This means $H$ lacks the \emph{ellipticity} that is assumed in \cite{3,4,bal2022mathematical,Drouot,QB}.
Indeed,
the term $(D_y - A_2 (x)) \sigma_2$ has (Weyl) symbol $(\zeta - A_2 (x)) \sigma_2$, which can remain small even when $|\zeta|$ and $|x|$ are large.


Observe that $H^2 = D^2_x + (D_y - A_2 (x))^2 + V(x)D_x \sigma_1 + V(x) (D_y -A_2(x)) \sigma_2 + R$ for $R$ bounded.
The leading-order terms $D^2_x + (D_y - A_2 (x))^2$ are precicely
the \emph{Iwatsuka Hamiltonian} analyzed in \cite{dombrowski2011quantization}, where there are many results analogous to the ones in this paper. Bulk invariants for magnetic Schr\"odinger operators are proposed and analyzed in \cite{ASS94,bellissard1994noncommutative}.
The distinguishing features of our setting are the additional domain walls $m$ and $V$, the lack of definiteness of $H$ (the spectrum of $H$ is not bounded above or below), and the fact that $H$ is a first-order matrix valued differential operator (instead of a second-order scalar one).
As we will see below, the lack of definiteness perhaps provides the biggest challenge; it will be easy to estimate the absolute value of spectral branches of $H$, but obtaining the sign of these branches will require more care.

In \cite{combes2005edge}, the interface conductivity is calculated for magnetic Schr\"odinger Hamiltonians 
with constant magnetic field and confining potentials. 
The interface conductivity is induced by the potentials much like an ``edge conductivity'' would be generated by ``hard wall'' Dirichlet boundary conditions.
A bulk-edge correspondence (involving the edge conductivity) for magnetic Dirac models is proven in \cite{cornean2022bulk}.

A main result in this paper is to derive an explicit expression for $2\pi \sigma_I (H)$ by means of a spectral flow; see Theorems \ref{prop:sf0} and \ref{prop:sf}  below. We show that $2\pi \sigma_I (H)$ is quantized and (for positive energies and constant $m$ and $V$) decreases in uniform increments as the strength of the magnetic field increases. 
This behavior of the interface conductivity resembles 
the integer quantum Hall effect. 
Theorem \ref{prop:sf0} 
is a bulk-interface correspondence in that it
relates the spectral flow of $H$ to a difference of bulk quantities. 
The terms in the difference are not expressed as bulk invariants, such as for instance Chern numbers as introduced in \cite{bellissard1994noncommutative,cornean2022bulk}. This issue, addressed for non-magnetic Dirac Hamiltonians in \cite{2,3}, is not considered further here.
For existing results on the bulk-interface correspondence involving a spectral flow, see also \cite{2,fukui2012bulk}.\\

Below is a brief summary of this paper.
We 
first define the spectral flow, 
and in doing so give an outline of Section \ref{sec:spec}.
We refer 
to \cite{phillips1996self} for a more generally applicable definition.
The set $\{\hat{H} (\zeta) : \zeta \in \mathbb{R} \}$ is a one-parameter family of self-adjoint operators,
with $\hat{H} (\zeta)$ holomorphic in $\zeta \in \mathbb{C}$
(see \cite[Chapter VII.1.1]{kato2013perturbation} for a precise definition of holomorphic operators).
We will show with Lemmas \ref{lemma:spectrumhat} and \ref{lemma:muanalytic} that the spectrum of $\hat{H} (\zeta)$ consists entirely of eigenvalues $\{\mu_j (\zeta)\}_{j \in \mathbb{Z}}$, where each $\mu_j : \mathbb{R} \rightarrow \mathbb{R}$ is analytic.
It turns out that for any $\alpha \in \rho (H_+) \cap \rho (H_-)$, there exists $\zeta_\alpha > 0$ such that no branches attain the value $\alpha$ when $|\zeta| > \zeta_\alpha$ (Lemmas \ref{lemma:mulim} and \ref{lemma:muInfty}).
Moreover, the number of branches to ever attain the value $\alpha$ is finite (Lemma \ref{lemma:finite}).

This means we can define the \emph{spectral flow of $H$ through $\alpha$}, denoted $\text{SF} (\alpha)$, to be the signed number of crossings of branches through $\alpha$.
That is, $\text{SF} (\alpha) := N_\uparrow - N_\downarrow$, where $N_\uparrow$ (resp. $N_\downarrow$) is the number of branches that are less than $\alpha$ when $\zeta < -\zeta_\alpha$ and greater than $\alpha$ when $\zeta > \zeta_\alpha$ (resp. greater than $\alpha$ when $\zeta < -\zeta_\alpha$ and less than $\alpha$ when $\zeta > \zeta_\alpha$).
With Theorem \ref{prop:sf0} we compute the spectral flow using the max-min principle \cite{RS4, Teschl} and perturbation theory \cite{kato2013perturbation}. 


In Section \ref{sec:stability}, we relate the interface conductivity to the spectral flow (Theorem \ref{prop:sf}) and prove its stability under a large class of perturbations (Theorems \ref{prop:stabilitysmall}-- \ref{prop:stabilitystrip}). The stability results are proved using pseudo-differential calculus as in \cite{3,4,bal2022mathematical,Drouot,QB}.
But since $H$ is not elliptic,
much of the pseudo-differential operator theory 
used in this existing literature 
does not apply.
Still, we are able to use Beals's criterion \cite[Proposition 8.3]{DS} and specific properties of $H$ (as a first-order differential operator with the spectrum of $H_\pm$ known) to obtain the necessary decay properties; see Lemmas \ref{lemma:resolvent} and \ref{lemma:decayPhiH}.
Similar stability results for (also non-elliptic) magnetic Schr\"odinger operators can be found in \cite{combes2005edge,dombrowski2011quantization}.

Appendix \ref{sec:PsiDO} contains the pseudo-differential calculus (and Helffer-Sj\"ostrand formula) that is needed for Section \ref{sec:stability}.

\section{Spectral analysis}\label{sec:spec}
In this section we calculate the spectral flow $\text{SF} (H; \alpha)$ for $\alpha \in \rho (H_+) \cap \rho (H_-)$.
This involves analyzing the limiting behavior of the branches of spectrum $\mu_j (\zeta)$ of $\hat{H} (\zeta)$ as $|\zeta| \rightarrow \infty$.
The multisets
$S_\pm := \{\lim_{\zeta \rightarrow \pm \infty} \mu_j (\zeta)\}_{j\in \mathbb{Z}}$
are determined using a standard max-min argument. Perturbation theory is then used to match elements in $S_+$ with those in $S_-$ to determine the quantities $\lim_{\zeta \rightarrow \pm \infty} \mu_j (\zeta)$ for every $j$.
We now state the main result, which can be interpreted as a bulk-interface correspondence.
\begin{theorem}\label{prop:sf0}
Fix $\alpha \in \rho (H_+) \cap \rho (H_-)$. Then $\text{SF} (H;\alpha)= I (H_-;\alpha) -I(H_+;\alpha)$, where
$$I (H_\pm; \alpha) = \sgn (B_\pm)\sgn (\alpha - V_\pm - m_\pm\sgn (B_\pm)) \Big(N(H_\pm;\alpha) + \frac{1}{2}\Big)$$ and \begin{align*}
    N(H_\pm;\alpha)=
    \begin{cases}
    0; & |\alpha-V_\pm|< \sqrt{2|B_\pm| + m_\pm^2},\\
    k; & \sqrt{2 k |B_\pm| + m_\pm^2} < |\alpha-V_\pm| < \sqrt{2 (k+1) |B_\pm| + m_\pm^2}, \quad k \in \mathbb{N}_+.
    \end{cases}
\end{align*}
\end{theorem}
Note that $N(H_\pm; \alpha)$ counts the ($\zeta-$independent) number of eigenvalues of $\hat{H}_\pm (\zeta) - V_\pm$ in the interval $(|m_\pm|,|\alpha-V_\pm|)$.
The above values of $\alpha$ for which $\text{SF} (\alpha)$ is not defined are precisely the elements of $\sigma (H_+) \cup \sigma (H_-)$.
Although the above expresses the (integer-valued) spectral flow as a difference of two bulk quantities,  it is unclear whether each bulk quantity may be interpreted as an invariant as in \cite{ASS94,bellissard1994noncommutative}. 

We refer 
to Figures \ref{fig:sf0} and \ref{fig:sf} for an illustration of the branches of spectrum of $H$, under various choices of the parameters $B_\pm, m_\pm$ and $V_\pm$. \footnote{These figures were generated using a standard finite difference approximation of the Hamiltonian with periodic boundary conditions. The periodic Hamiltonian also has spurious eigenvalues (corresponding to eigenfunctions that are localized near the boundary; see \cite{QB} for more details), which are not included in our plots.} For example, in the top right panel of Figure \ref{fig:sf} we see that $\text{SF} (H; \alpha) = 1$ for all $\alpha$ near $0$, while $\text{SF} (H; \alpha) = -1$ for $\alpha$ near $2.5$.\\\\


\begin{figure}[t]
\centering
  \includegraphics [scale=.20] {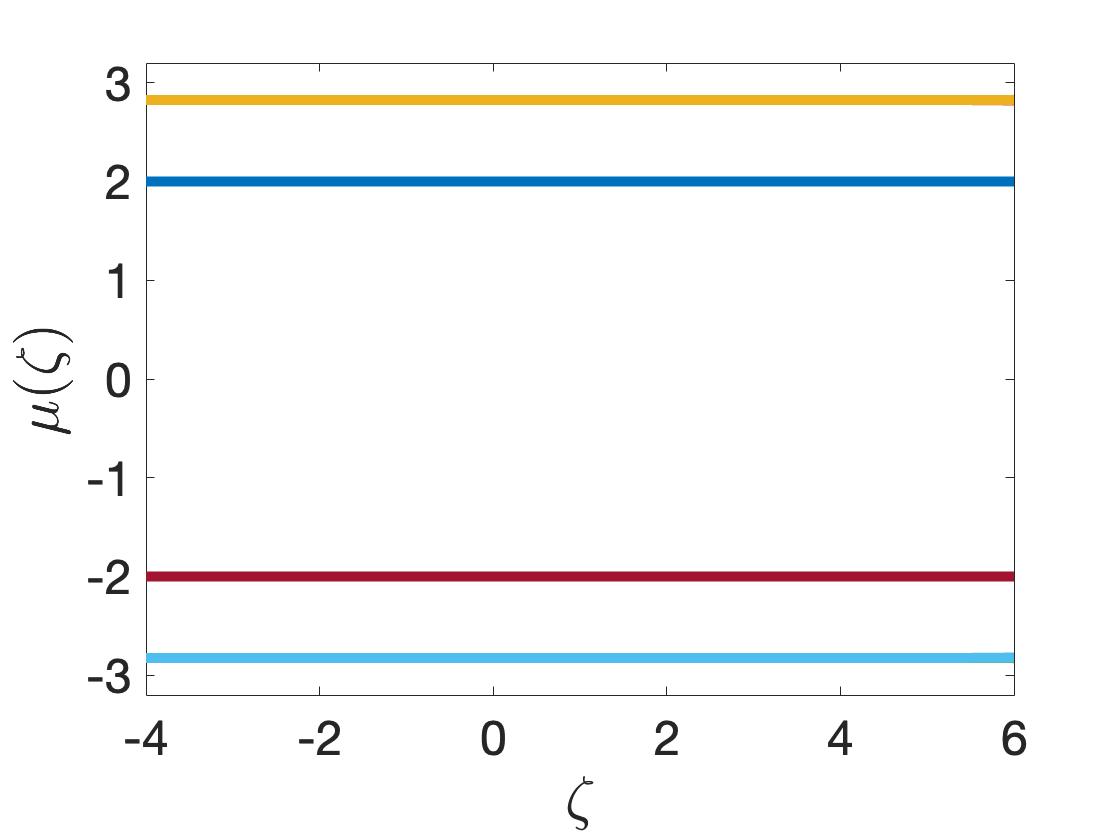}
  \includegraphics [scale=.20] {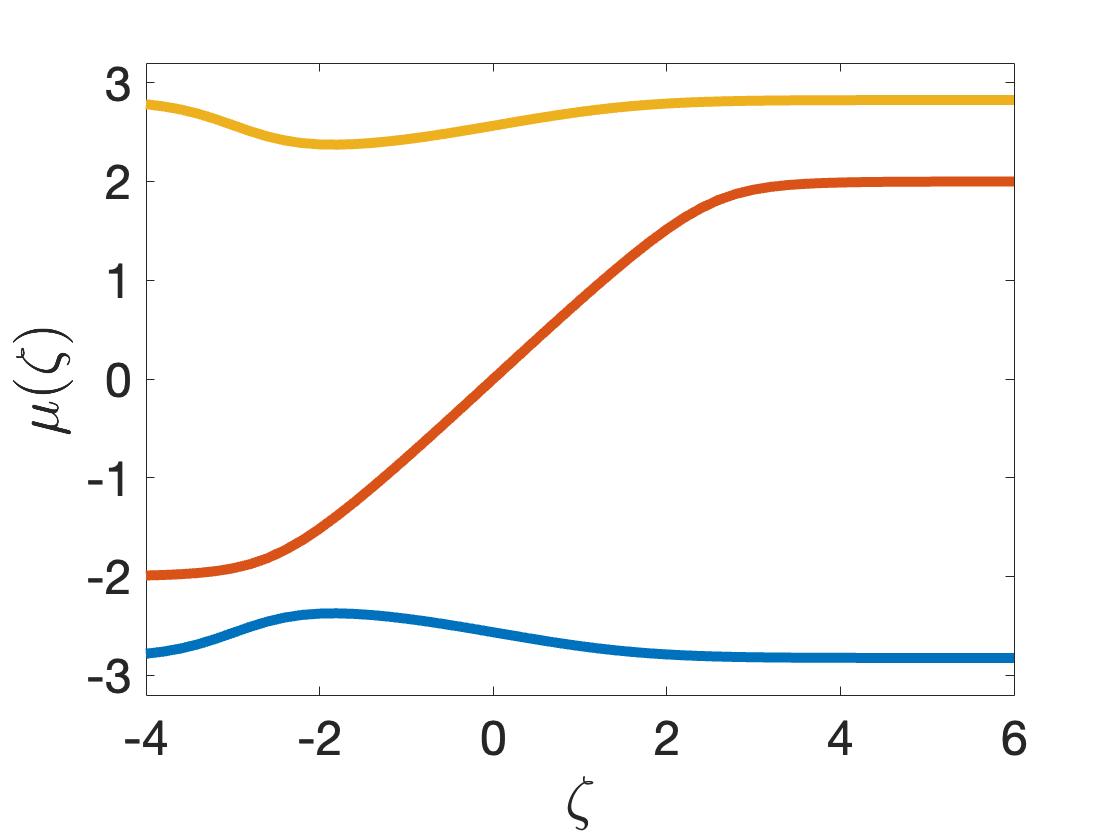}\\
  \includegraphics [scale=.20] {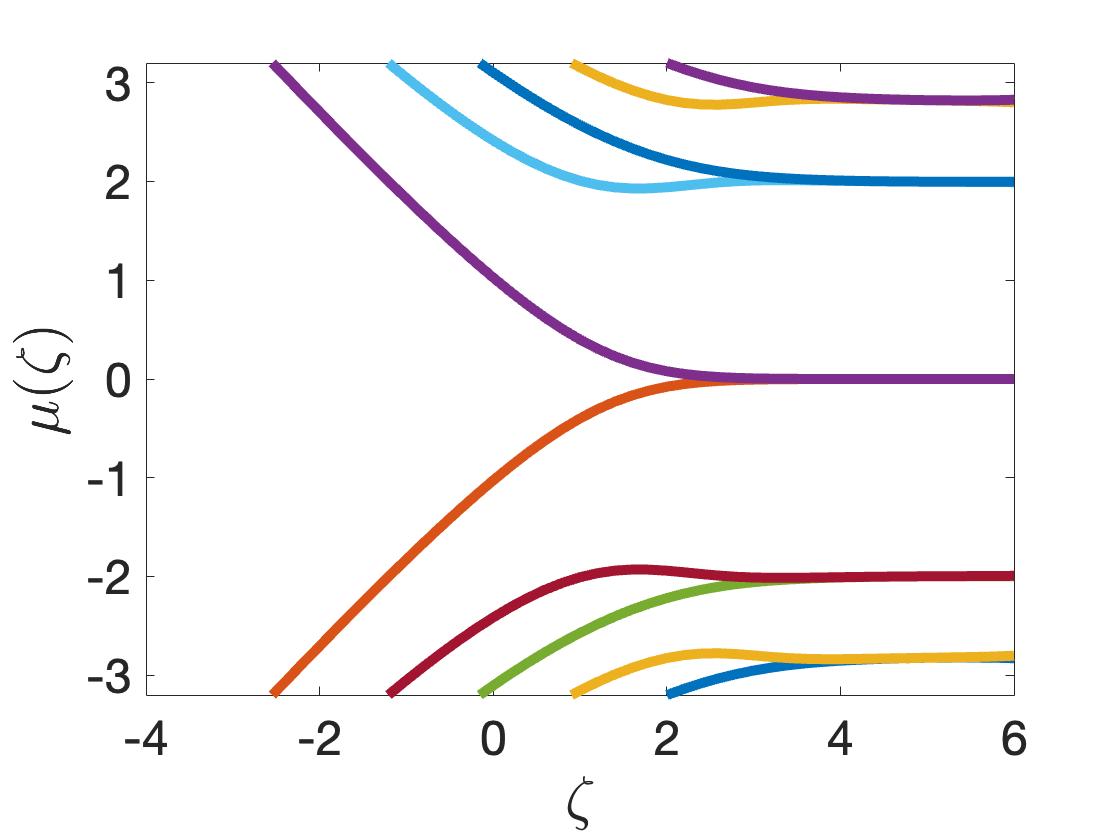}
  \includegraphics [scale=.20] {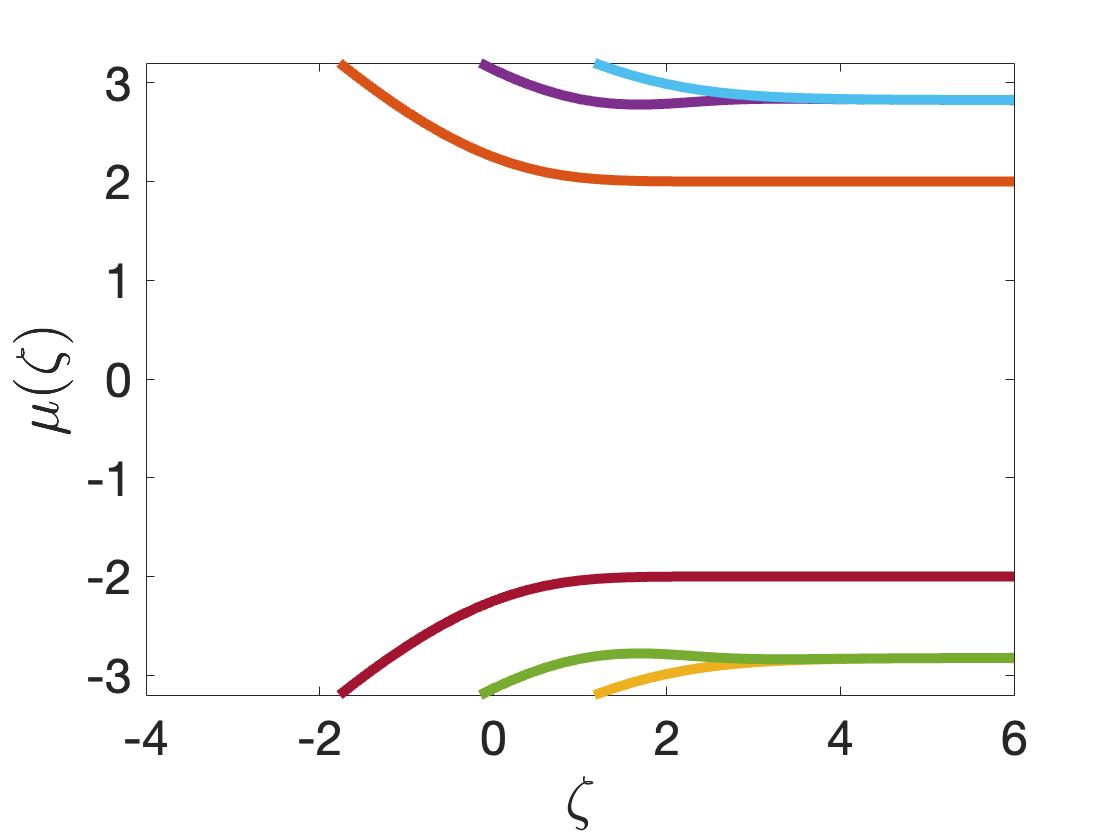}
\caption{Plots of the lowest-magnitude eigenvalues of $\hat{H} (\zeta)$ as a function of $\zeta$, for different combinations of domain walls in $B$ and $m$. For all plots, $V \equiv 0$. 
The choices of parameters are $B \equiv 2$ and $m \equiv 2$ (top left); $B \equiv 2$ and $m_+ = 2 = -m_-$ (top right); $B_+ = 2 = -B_-$ and $m \equiv 0$ (bottom left); $B_+ = 2 = -B_-$ and $m \equiv 2$ (bottom right).
As $\zeta \rightarrow +\infty$, these curves converge to 
elements of $\sigma (\hat{H}_+ (\zeta)) \cup \sigma (\hat{H}_+ (\zeta))$ as predicted by the theory; see Lemma \ref{lemma:muInfty}.
The top left panel illustrates the spectrum of a bulk Hamiltonian; see Lemma \ref{lemma:spectrumhat}. As demonstrated by the top right plot, a transition in the sign of $m$ 
generates a nonzero spectral flow for $\alpha$ near $0$. Comparing the two bottom plots, we see how a constant nonzero $m$ opens a gap at $0$.
Although it may look like 
the eigenvalues in the bottom plots are degenerate for large $\zeta$, this is not the case (Lemma \ref{lemma:simple}). They only have the same limit as $\zeta \rightarrow \infty$. 
}
\label{fig:sf0}
\end{figure}

\begin{figure}[t]
\centering
  \includegraphics [scale=.20] {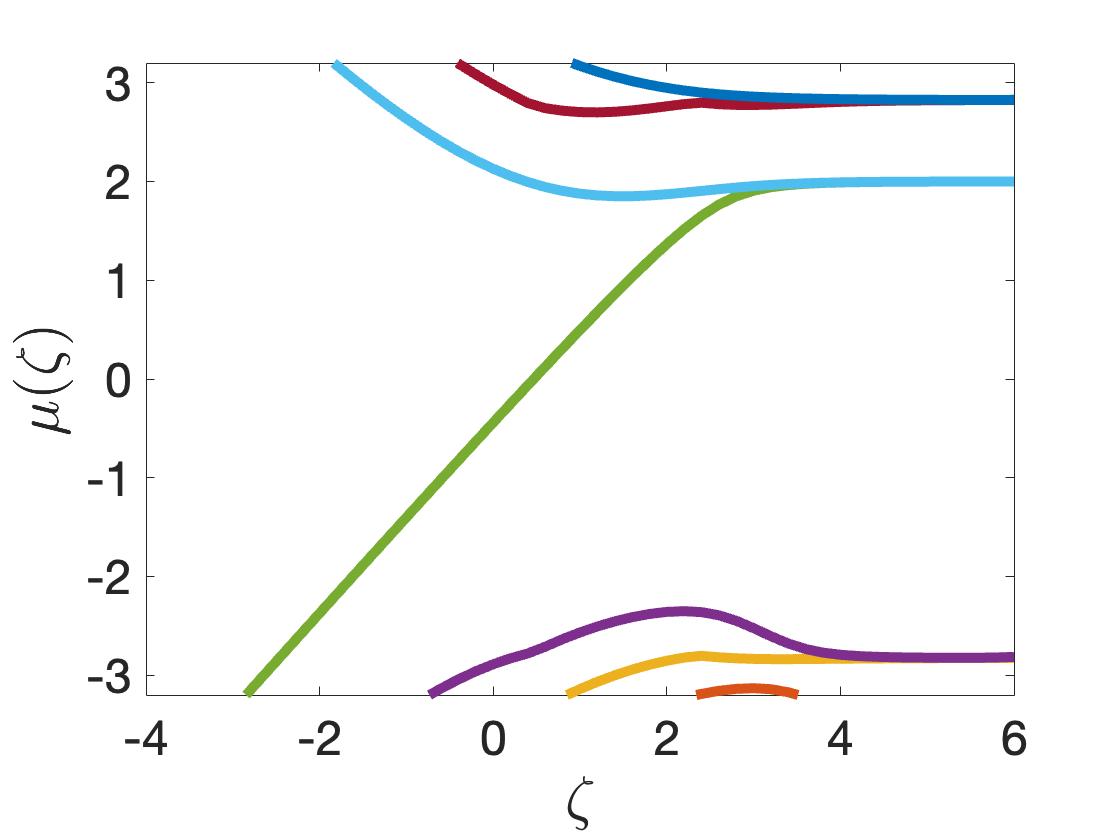}
  \includegraphics [scale=.20] {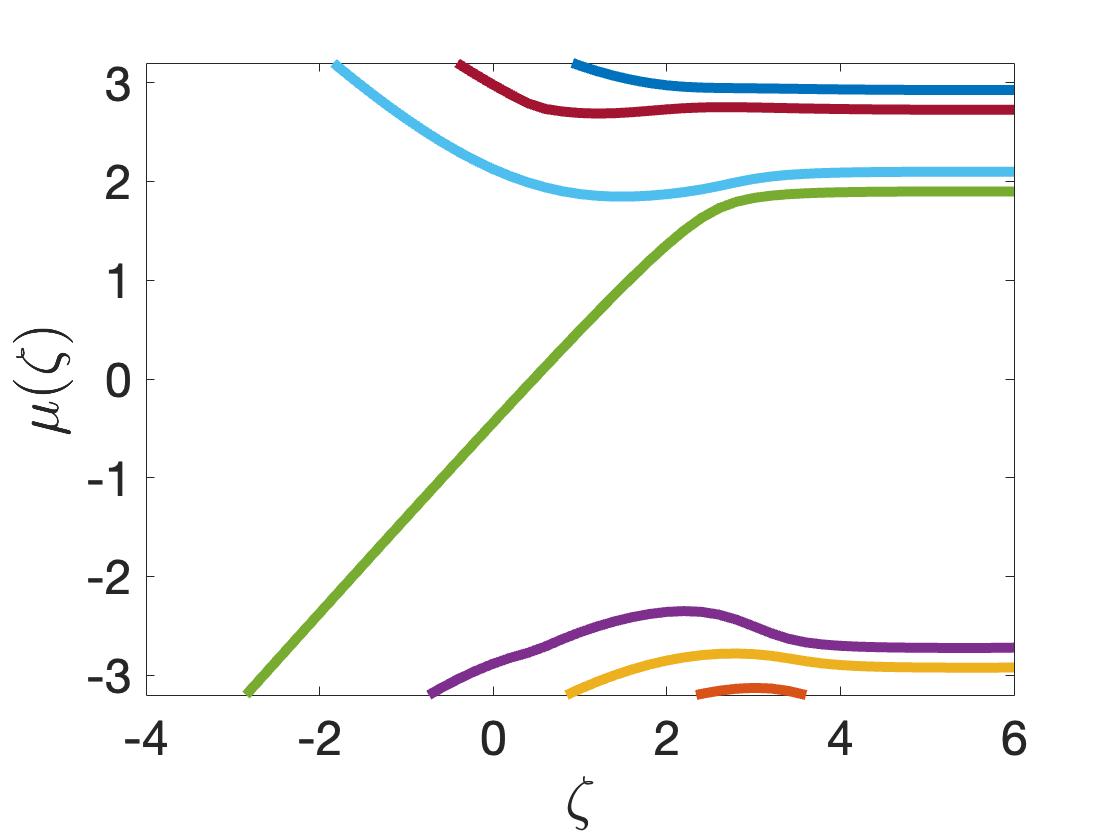}\\
  \includegraphics [scale=.20] {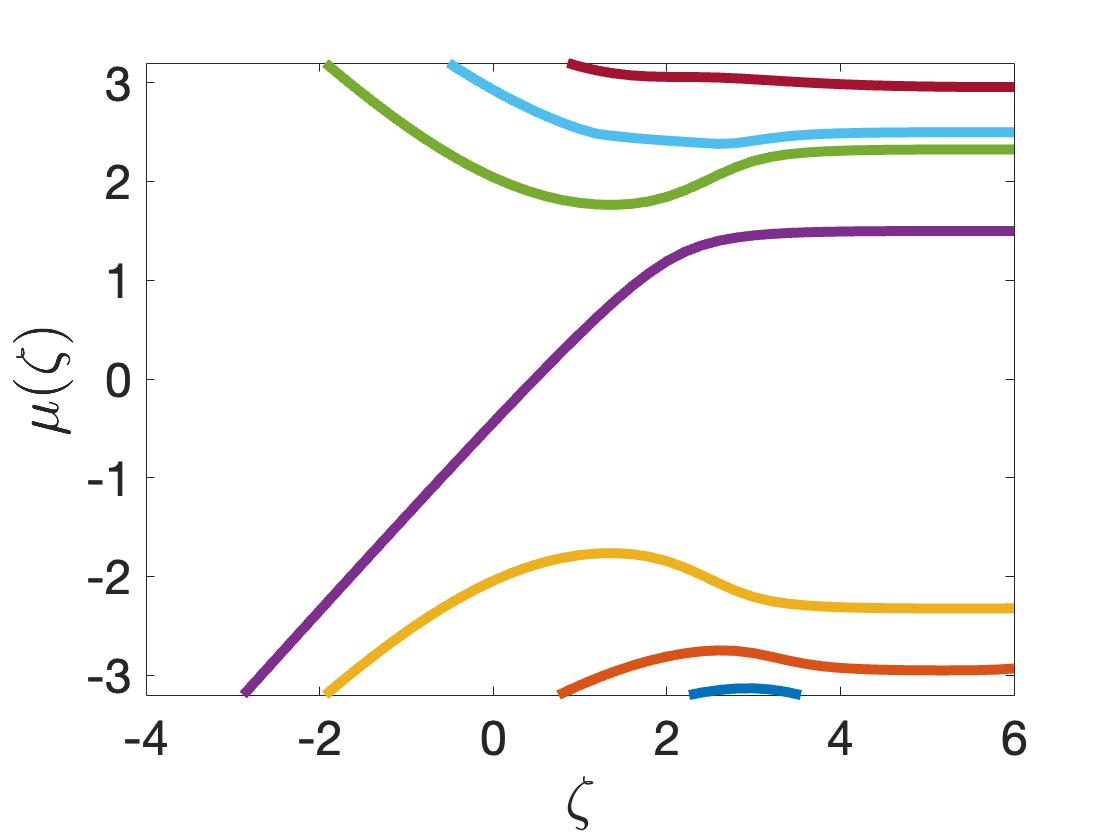}
  \includegraphics [scale=.20] {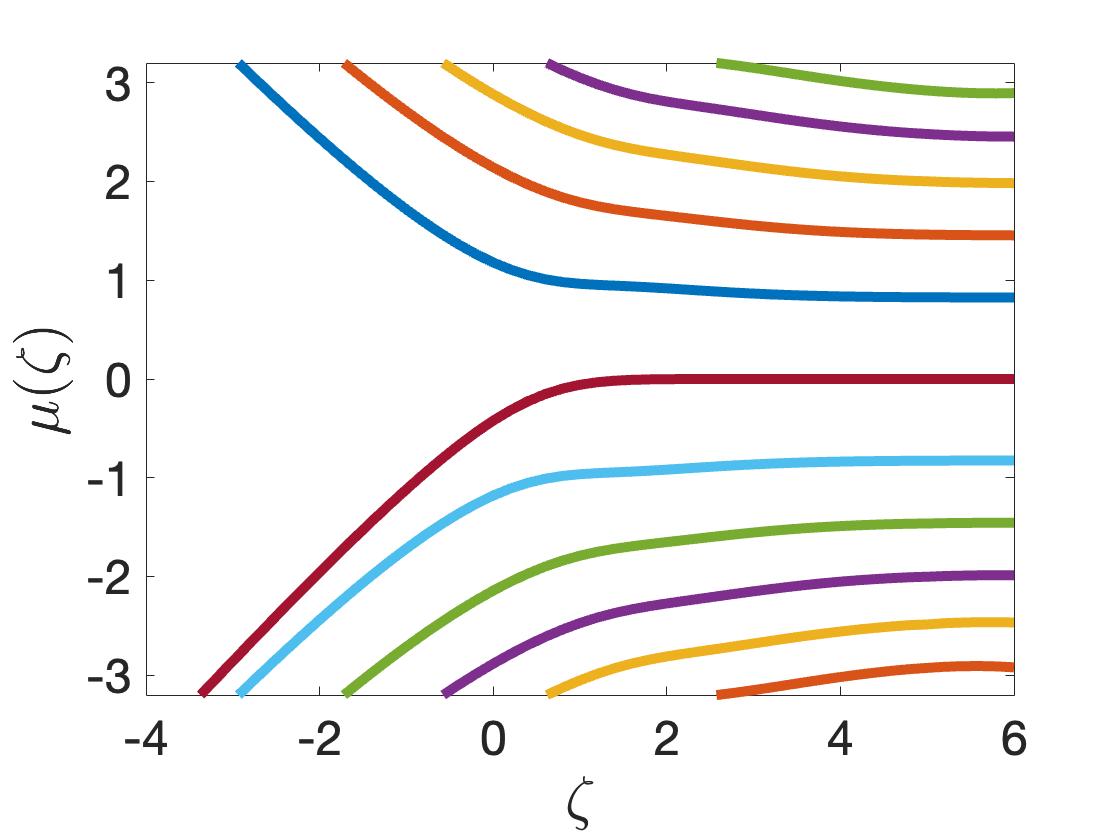}
\caption{Effect of the electric potential $V$ on the
lowest-magnitude eigenvalues of $\hat{H} (\zeta)$. 
All plots have $B_+ = 2 = -B_-$ and $m_+ = 2 = -m_-$. 
The choices of electric potential are
$V \equiv 0$ (top left); $V_+ = 0.1 = -V_-$ (top right); $V_+ = 0.5 = -V_-$ (bottom left); $V_+ = 2 = -V_-$ (bottom right). 
Notice that, as predicted by Lemma \ref{lemma:spectrumhat}, a small domain wall in $V$ 
opens a gap between pairs of branches for large $\zeta$ (top panels).}
\label{fig:sf}
\end{figure}

The rest of this section is devoted to proving Theorem \ref{prop:sf0} (and the statements we made before it). 
We begin by determining the spectrum of $\hat{H}_\pm (\zeta)$.
\begin{lemma}\label{lemma:spectrumhat}
For any $\zeta \in \mathbb{R}$, the spectrum of $\hat{H}_\pm (\zeta)$ consists entirely of eigenvalues and is given by
\begin{align}\label{eq:spectrumhat}
    \sigma (\hat{H}_\pm (\zeta)) = \{\eps \sqrt{2 k |B_\pm| + m_\pm^2}+V_\pm : \eps \in \{-1,1\}, k \in \mathbb{N}_+\} \bigcup \{m_\pm \sgn (B_\pm) + V_\pm \}.
\end{align}
\end{lemma}
Above, the subscripts $\pm$ (and one operation $\mp$) are understood to correspond to the operator $\hat{H}_\pm$ in question.
Note that the spectrum of $\hat{H}_\pm (\zeta)$ is independent of $\zeta$; thus the spectrum of $H_\pm$ is also given by \eqref{eq:spectrumhat}, only now these values all belong to the \emph{essential spectrum}.
\begin{proof}
Suppose for concreteness that $B_+ > 0$.
Set $V_+ = 0$ without loss of generality, as this term only contributes a uniform shift of the spectrum.
Any eigenpair $\mu$ and $\psi$ of $\hat{H}_+ (\zeta)$ must satisfy $\hat{H}^2_+ (\zeta) \psi = \mu^2 \psi$, with $\hat{H}^2_+ (\zeta) = D_x^2 + (\zeta - B_+ x)^2 + m_+^2 - B_+ \sigma_3.$
The eigenelements of $\hat{H}^2_+ (\zeta)$ are well known. The eigenvalues are $\nu_k = 2k|B_+| + m_+^2$ for $k \in \mathbb{N}$. When $k \in \mathbb{N}_+$, $\nu_k$ has multiplicity $2$ and the eigenfunctions are $\psi_{k,\uparrow} = (\phi_{k}, 0)$ and $\psi_{k,\downarrow}=(0,\phi_{k-1})$, where $\phi_k (x) = |B_+|^{-1/4} \tilde{\phi}_k (\sqrt{|B_+|}x - \zeta)$ with $\tilde{\phi}_k$ the Hermite functions.
We see that $\nu_0$ has multiplicity one with eigenfunction $\psi_0 = (\phi_0,0)$.
We then verify that the eigenvalues of $\hat{H}_+ (\zeta)$ are $\mu_{k,\eps} = \eps \sqrt{\nu_k}$ for $\eps \in \{-1,1\}$, $k \in \mathbb{N}_+$ and $\mu_0 = m_+$, with corresponding eigenfunctions $\psi_{k,\eps} = c_{1,\eps} \psi_{k,\uparrow} + c_{2,\eps} \psi_{k,\downarrow}$ and $\psi_0$ for some $|c_{1,\eps}|^2 + |c_{2,\eps}|^2 = 1$.
If instead $B_+ < 0$, we would instead have $\psi_{k,\uparrow} = (\phi_{k-1}, 0)$ and $\psi_{k,\downarrow}=(0,\phi_{k})$. All elements of the spectrum would be the same as before with the exception of $\mu_0= -m_+$ now.
The eigenelements of $\hat{H}_- (\zeta)$ are calculated similarly.
This completes the result.
\end{proof}

\begin{lemma}\label{lemma:simple}
For any $\zeta \in \mathbb{R}$, the spectrum of $\hat{H}(\zeta)$ consists entirely of simple eigenvalues.
\end{lemma}

\begin{proof}
Fix $\zeta \in \mathbb{R}$. 
The Weyl symbol (see Appendix \ref{sec:PsiDO} for the definition) of $\hat{H} (\zeta)$ grows linearly in $\aver{x, \xi}$, as demonstrated by \eqref{eq:ell} below.
This implies that
$(i + \hat{H} (\zeta))^{-1}$ is compact (see e.g. \cite{4,Bony}), and hence
$\sigma (\hat{H} (\zeta))$ consists entirely of eigenvalues.
We now prove that each eigenvalue has multiplicity one.
Suppose $\psi$ and $\phi$ are eigenfunctions of $\hat{H}$ corresponding to eigenvalue $\mu$.
This means
\begin{align*}
    -i \psi_2' - i (\zeta - A_2) \psi_2 + (V+m) \psi_1 &= \mu \psi_1\\
    -i\psi_1' + i (\zeta - A_2) \psi_1 +(V-m) \psi_2 &= \mu \psi_2,
\end{align*}
with the same equations holding also with $\psi$ replaced by $\phi$.
It follows that
\begin{align*}
    \psi_1 \phi_2' + \psi_1' \phi_2 = i(\mu - m - V) \psi_1 \phi_1 + i(\mu+m-V)\psi_2 \phi_2 = \psi_2 \phi_1' + \psi_2' \phi_1,
\end{align*}
so that
\begin{align*}
    \partial_x (\psi_1 \phi_2 - \psi_2 \phi_1) = \psi_1 \phi_2' + \psi_1 ' \phi_2 - (\psi_2 \phi_1' + \psi_2' \phi_1) = 0.
\end{align*}
Since $\psi,\phi,\psi',\phi'$ all go to zero as $|x| \rightarrow \infty$, we have shown that
$\psi_1 \phi_2 - \psi_2 \phi_1 \equiv 0$.
Thus for every $x \in \mathbb{R}$, the vectors $(\psi_1 (x), \psi_2(x))$ and $(\phi_1 (x), \phi_2 (x))$ are linearly dependent.
Normalize the eigenfunctions so that $\psi(x_0) = \phi(x_0)$ for some $x_0 \in \mathbb{R}$ (this can be done because $\psi$ and $\phi$ are continuous, hence there must exist a point at which both functions are zero or both functions are nonzero).
Letting $\nu := \psi - \phi$, we see that
$$
    \nu '(x) = F(x) \nu(x), \quad x \in \mathbb{R}; \qquad\qquad \qquad \nu (x_0) = 0,
$$
for some $F \in \mathcal{C}^\infty (\mathbb{R}; \mathbb{C}^{2\times 2})$.
By regularity of $F$, the standard uniqueness result for first-order ODE implies that $\nu \equiv 0$.
This completes the result.
\end{proof}


\begin{lemma}\label{lemma:muanalytic}
There exists a countable collection of analytic functions, $\{\mu_j\}_{j \in \mathbb{Z}} \subset \mathcal{C}^\omega (\mathbb{R})$, such that for each $\zeta \in \mathbb{R}$, $\sigma (\hat{H} (\zeta)) = \{\mu_j (\zeta): j \in \mathbb{Z}\}$.
\end{lemma}
Here, $\mathcal{C}^\omega (\mathbb{R})$ denotes the set of analytic functions on $\mathbb{R}$.
\begin{proof}
Note that $\hat{H} (\zeta)$ is holomorphic in $\zeta \in \mathbb{C}$
(see \cite[Chapter VII.1.1]{kato2013perturbation} for the precise definition)
and self-adjoint whenever $\zeta \in \mathbb{R}$ \cite{4,Bony,thaller2013dirac}.
The result then
follows from Lemma \ref{lemma:simple} and \cite[Theorems VII.1.7 and VII.1.8]{kato2013perturbation}. 
\end{proof}


We now show that bounded perturbations of $\hat{H}$ cannot change the $\mu_j$ by too much.
\begin{lemma}\label{lemma:mulambda}
Let $m_j \in \fs (m_{j,-}, m_{j,+})$ and $V_j \in \fs (V_{j,-},V_{j,+})$ for $j \in \{1,2\}$, where  the $m_{j,\pm}$ and $V_{j,\pm}$ are real numbers.
For $\zeta \in \mathbb{R}$ and $\lambda \in [0,1]$, set $$\hat{H}(\zeta,\lambda) = D_x \sigma_1 + (\zeta - A_2(x)) \sigma_2 + m_1(x) \sigma_3 + V_1(x) \sigma_0 + \lambda ((m_2(x) - m_1(x)) \sigma_3 + (V_2(x)-V_1(x)) \sigma_0).$$
Let $\{\mu_j (\zeta,\lambda)\}_{j \in \mathbb{Z}}$ denote the eigenvalues of $\hat{H} (\zeta,\lambda)$. 
Then the $\mu_j$ can be chosen analytic in $(\zeta,\lambda)$ and
\begin{align*}
    |\partial_{\lambda} \mu_j (\zeta, \lambda)| \le 
    \norm{m_2 - m_1}_\infty + \norm{V_2-V_1}_\infty.
\end{align*}
\end{lemma}
\begin{proof}
Since $m$ and $V$ were arbitrary switch functions,
it follows from Lemma \ref{lemma:simple} and 
\cite[Theorems VII.1.7 and VII.1.8]{kato2013perturbation} as before that the eigenvalues and eigenprojections of $\hat{H} (\zeta, \lambda)$
are holomorphic in
$(\zeta,\lambda) \in \mathbb{C}^2$.
Since the eigenvalues are real whenever $\zeta$ and $\lambda$ are real, we have proved the first part of the lemma.
Letting $\mu (\zeta,\lambda)$ denote such an eigenvalue with $\Pi_{\zeta,\lambda} = \psi_{\zeta, \lambda}\psi^*_{\zeta, \lambda}$ the projection onto the corresponding (one-dimensional) eigenspace, we have
\begin{align*}
    \hat{H} (\zeta, \lambda) \psi_{\zeta, \lambda} = \mu (\zeta, \lambda) \psi_{\zeta, \lambda}.
\end{align*}
Now fix $\zeta \in \mathbb{R}$ and $\lambda \in [0,1]$, and
let $h > 0$. 
Evaluating the difference of the above between $\lambda$ and $\lambda + h$ yields
\begin{align*}
(\hat{H} (\zeta, \lambda &+ h) - \hat{H} (\zeta, \lambda)) \psi_{\zeta,\lambda+h} + \hat{H} (\zeta, \lambda) (\psi_{\zeta, \lambda + h} - \psi_{\zeta, \lambda})\\ &= 
    \hat{H} (\zeta, \lambda + h) \psi_{\zeta, \lambda +h} - \hat{H} (\zeta, \lambda) \psi_{\zeta, \lambda} 
    = \mu (\zeta, \lambda + h) \psi_{\zeta, \lambda + h} - \mu (\zeta, \lambda) \psi_{\zeta, \lambda}\\
    &= (\mu(\zeta, \lambda + h) - \mu (\zeta, \lambda)) \psi_{\zeta,\lambda + h} + \mu (\zeta, \lambda) (\psi_{\zeta, \lambda + h} - \psi_{\zeta, \lambda}).
\end{align*}
Multiplying both sides by $\bar{\psi}_{\zeta,\lambda}$ and taking inner products, we obtain
\begin{align*}
    (\psi_{\zeta,\lambda}, (\hat{H} (\zeta, \lambda + h) &- \hat{H} (\zeta, \lambda)) \psi_{\zeta,\lambda+h}) = (\mu(\zeta, \lambda + h) - \mu (\zeta, \lambda))(\psi_{\zeta,\lambda},\psi_{\zeta, \lambda + h}).
\end{align*}
Using that $\hat{H} (\zeta, \lambda + h) - \hat{H} (\zeta, \lambda) = h((m_2(x) - m_1(x)) \sigma_3 + (V_2 (x) - V_1 (x)) \sigma_0)$,
we divide both sides by $h$ to get
\begin{align}\label{eq:muprime}
\begin{split}
    (\psi_{\zeta,\lambda}, ((m_2(x) - m_1(x)) \sigma_3 + (V_2 (x) &- V_1 (x)) \sigma_0) \psi_{\zeta,\lambda+h})\\
    &= h^{-1}(\mu(\zeta, \lambda + h) - \mu (\zeta, \lambda))(\psi_{\zeta,\lambda},\psi_{\zeta, \lambda + h}).
    \end{split}
\end{align}
Since $\Pi_{\zeta,\lambda}$ is holomorphic in $\lambda$, 
the operator norm $\norm{\Pi_{\zeta, \lambda + h} - \Pi_{\zeta,\lambda}}$ goes to zero as $h \rightarrow 0$.
Hence $\norm{(\psi_{\zeta,\lambda+h},\psi_{\zeta,\lambda}) \psi_{\zeta,\lambda+h} - \psi_{\zeta,\lambda}} \rightarrow 0$, meaning that $|(\psi_{\zeta,\lambda+h},\psi_{\zeta,\lambda})| \rightarrow 1$.
Taking the absolute value of \eqref{eq:muprime} and sending $h \rightarrow 0$, we thus obtain
\begin{align*}
    |\partial_{\lambda} \mu (\zeta, \lambda)| \le 
    \norm{m_2 - m_1}_\infty + \norm{V_2-V_1}_\infty.
\end{align*}
This completes the proof.
\end{proof}

Next we determine the multisets $\{\lim_{\zeta \rightarrow \pm \infty} \mu_j (\zeta)\}_{j \in \mathbb{Z}}$, 
which depends on the signs of $B_+$ and $B_-$. 
Given two sets $A$ and $B$, we use the notation $A+B$ to denote the multiset formed by combining $A$ and $B$. That is, $x \in A+B$ if and only if $x \in A\cup B$, where the multiplicity of $x$ is $2$ if $x \in A \cap B$ and $1$ otherwise.
Although each $\mu_j$ has multiplicity $1$ (Lemma \ref{lemma:simple}), it is possible that two distinct $\mu_j$ converge to the same value (as $\zeta \rightarrow \infty$ if $B_- < 0 < B_+$; as $\zeta \rightarrow -\infty$ if $B_+ < 0 < B_-$).
\begin{lemma}\label{lemma:mulim} 
\begin{enumerate}
\item \label{it:m0p} Suppose $B_- < 0 < B_+$. Then for each $j\in \mathbb{Z}$, $\mu_{j,\infty} := \lim_{\zeta \rightarrow \infty} \mu_j (\zeta)$ exists and belongs to $\sigma (\hat{H}_+ (\zeta)) \cup \sigma (\hat{H}_- (\zeta))$.
For each $\nu \in \sigma (\hat{H}_+ (\zeta)) + \sigma (\hat{H}_- (\zeta))$, there exists exactly one index $j \in \mathbb{Z}$ such that $\mu_{j,\infty}= \nu$.
\item \label{it:p0m} Suppose $B_+ < 0 < B_-$. 
Then for each $j\in \mathbb{Z}$, $\mu_{j,-\infty} := \lim_{\zeta \rightarrow -\infty} \mu_j (\zeta)$ exists and belongs to $\sigma (\hat{H}_+ (\zeta)) \cup \sigma (\hat{H}_- (\zeta))$.
For each $\nu \in \sigma (\hat{H}_+ (\zeta)) + \sigma (\hat{H}_- (\zeta))$, there exists exactly one index $j \in \mathbb{Z}$ such that $\mu_{j,-\infty}= \nu$.
\item \label{it:0pm} Suppose 
$0 < B_+, B_-$. Then for each $j \in \mathbb{Z}$, $\mu_{j,\pm\infty} := \lim_{\zeta \rightarrow \pm\infty} \mu_j (\zeta)$ exists and belongs to $\sigma (\hat{H}_\pm (\zeta))$. 
For each $\nu \in \sigma (\hat{H}_\pm (\zeta))$, there exists exactly one index $j \in \mathbb{Z}$ such that $\mu_{j,\pm\infty}= \nu$.
\item \label{it:pm0} Suppose
$B_+, B_- < 0$. Then for each $j \in \mathbb{Z}$, $\mu_{j,\pm\infty} := \lim_{\zeta \rightarrow \pm\infty} \mu_j (\zeta)$ exists and belongs to $\sigma (\hat{H}_\mp (\zeta))$.
For each $\nu \in \sigma (\hat{H}_\mp (\zeta))$, there exists exactly one index $j \in \mathbb{Z}$ such that $\mu_{j,\pm\infty}= \nu$.
\end{enumerate}
\end{lemma}
\begin{proof}
Suppose $B_- < 0 < B_+$. 
Let $\{\nu_{k, \pm}: k \in \mathbb{N}\}$ denote the full set of eigenvalues of $\hat{H}_\pm$.
Applying the max-min principle \cite{RS4, Teschl} to $\hat{H}^2 (\zeta)$ (see \cite{cycon2009schrodinger,QB,yafaev2008spectral} for similar arguments), it follows that 
there is a bijection $\iota : \mathbb{Z} \rightarrow \mathbb{N} \times \{+, -\}$ such that for every $j \in \mathbb{Z}$, $\mu_j^2 (\zeta) \rightarrow \nu^2_{\iota (j)}$ as $\zeta \rightarrow \infty$.

Now fix 
$j \in \mathbb{Z}$ and
let $K_j := \{i \in \mathbb{Z} : \nu^2_{\iota (i)} = \nu^2_{\iota (j)}\}$.
(Lemma \ref{lemma:spectrumhat} implies that $K_j$ can have at most four elements.)
By Lemma \ref{lemma:mulambda}, the quantities
$\lim_{\zeta \rightarrow \infty} \mu_j^2 (\zeta)$ are continuous in shifts of $V$. Thus
there is a bijection $\iota_j : K_j \rightarrow \iota (K_j)$ such that
for all $i \in K_j$ and
$\alpha$ sufficiently small, $(\mu_i (\zeta) + \alpha)^2 \rightarrow (\nu_{\iota_j (i)} + \alpha)^2$ as $\zeta \rightarrow \infty$.
Indeed, the previous paragraph applies also to $H+\alpha$, as $V \in \fs (V_-, V_+)$ was arbitrary.
Hence $\lim_{\zeta \rightarrow \infty} \mu_i (\zeta) = \nu_{\iota_j (i)}$ for all $i \in K_j$, as desired.

The proofs 
for cases \ref{it:p0m}--\ref{it:pm0} are similar.
\end{proof}

\begin{lemma}\label{lemma:muInfty}
\begin{enumerate}
\item Suppose $B_- < 0 < B_+$. For every $M>0$, there exists $\zeta_0 \in \mathbb{R}$ such that $|\mu_j (\zeta)| > M$ for all $\zeta < \zeta_0$ and $j \in \mathbb{Z}$.
\item Suppose $B_+ < 0 < B_-$. For every $M>0$, there exists $\zeta_0 \in \mathbb{R}$ such that $|\mu_j (\zeta)| > M$ for all $\zeta > \zeta_0$ and $j \in \mathbb{Z}$.
\end{enumerate}
\end{lemma}
By continuity of the $\mu_j$, Lemma \ref{lemma:muInfty} implies that when $\pm B_- < 0 < \pm B_+$, each $\mu_j$ goes either to $+\infty$ or $-\infty$ as $\zeta \rightarrow \mp\infty$.
\begin{proof}
Suppose $B_-<0<B_+$.
We see that
\begin{align}\label{eq:square}
\begin{split}
    \hat{H}^2 (\zeta) = D^2_x + (\zeta - A_2 (x))^2 &+ m^2 (x) - A'_2(x) \sigma_3 - m'(x) \sigma_2 + V^2 (x)\\ &+ 2V(x) (D_x \sigma_1 + (\zeta - A_2(x)) \sigma_2 + m(x) \sigma_3) - i V'(x)\sigma_1,
    \end{split}
\end{align}
where $\sigma_0$ has been dropped from the notation.
Since $V$ is bounded, we know there exists a constant $C>0$ such that
\begin{align*}
    |(\psi, V D_x \sigma_1 \psi)| \le C \norm{\psi} \norm{\psi '}
\end{align*}
for all $\psi \in \mathcal{H}^1$.
Since $(\psi, D_x^2 \psi) = \norm{\psi '}^2$, it follows that the operator $D^2_x + V(x) D_x \sigma_1$ is bounded from below.
The function $A_2$ is also bounded from below, hence
\begin{align*}
    \lim_{\zeta \rightarrow -\infty}\; \inf_{x \in \mathbb{R}} (\zeta - A_2 (x))^2 = \infty.
\end{align*}
Since $m,V,m',V',A'_2$ are all bounded, we conclude that all eigenvalues of $\hat{H}^2 (\zeta)$ go to infinity uniformly as $\zeta \rightarrow -\infty$.
The proof for the case $B_+<0<B_-$ is similar.
This completes the result.
\end{proof}



We also need the following result, which ensures that the spectral flow is well defined.
\begin{lemma}\label{lemma:finite}
For every $\alpha \in \rho (H_+) \cap \rho (H_-)$, the set $\mathcal{T}_\alpha := \{j : \alpha \in \Ran (\mu_j)\}$ is finite.
\end{lemma}
Above, $\Ran (\mu_j) := \{\mu_j (\zeta) : \zeta \in \mathbb{R}\}$ is the \emph{range} (or image) of $\mu_j$.
\begin{proof}
Fix $\alpha \in \rho (H_+) \cap \rho (H_-)$.
The previous lemmas imply the existence of $\bar{\zeta} > 0$ such that $\mu_j (\zeta) \ne \alpha$ for all $j \in \mathbb{Z}$ and $|\zeta| > \bar{\zeta}$.
For $\zeta \in \mathbb{R}$ and $\beta >0$, let $N (\zeta,\beta) := |\{j : \mu^2_j (\zeta) < \beta \}|$.
Suppose by contradiction that $\mathcal{T}_\alpha$ is infinite. Then there exists a sequence $(\zeta_k) \subset [-\bar{\zeta},\bar{\zeta}]$ and a number $\zeta_* \in [-\bar{\zeta},\bar{\zeta}]$ such that $\zeta_k \rightarrow \zeta_*$ and $N(\zeta_k,\alpha^2) \rightarrow \infty$ as $k \rightarrow \infty$.
But Lemma \ref{lemma:simple} implies that $N(\zeta_*, \alpha^2 + 1) < \infty$, hence
there exist $i,j \in \mathbb{Z}$ such that $\mu_i (\zeta_*) \le - \sqrt{\alpha^2+1}$ and $\mu_j (\zeta_*) \ge \sqrt{\alpha^2 + 1}$.
Lemma \ref{lemma:simple} and
the fact that $N (\zeta_k, \alpha^2) \rightarrow \infty$ imply that for all $k$ sufficiently large, there exists $\ell \in \{i,j\}$ such that $\mu_\ell^2 (\zeta_k) < \alpha^2$.
Thus either $\mu_i$ or $\mu_j$ is not continuous at $\zeta_*$, which contradicts Lemma \ref{lemma:muanalytic}.
\end{proof}

It follows that for any $[E_1, E_2] \subset \rho (H_+) \cap \rho (H_-)$ and $\Phi \in \mathcal{C}^\infty_c (E_1, E_2)$, the kernel of $\Phi (H)$ is given by
\begin{align}\label{eq:kernelPhi}
    k_\Phi(x,x'; y-y') = \int_{\mathbb{R}} \sum_{j\in \mathcal{J}} \Phi(\mu_j(\zeta)) \psi_j(x, \zeta) \psi_j^*(x', \zeta) \frac{e^{i(y-y')\zeta}}{2\pi} d\zeta,
\end{align}
where $\{\psi_j (\cdot, \zeta)\}_{j \in \mathbb{Z}}$ denotes the normalized eigenfunctions of $\hat{H} (\zeta)$, and 
$\mathcal{J} \subset \mathbb{Z}$ is the finite set of indices corresponding to branches $\mu_j$ that ever enter the interval $[E_1, E_2]$.
Note that the above asymptotic analysis of $\mu_j$ implies that $\Phi (\mu_j (\zeta))$ vanishes whenever $|\zeta|$ is sufficiently large, hence the integral over $\mathbb{R}$ in \eqref{eq:kernelPhi} can be replaced by an integral over a bounded interval.

To get the spectral flow, it remains to determine the sign of $\mu_j (\zeta)$ as $\zeta \rightarrow -\infty$. We will do this first by imposing additional constraints on $\hat{H}$.
\begin{lemma}\label{lemma:gap}
Let $m_0 \in \mathbb{R}$ such that $m_0^2 > \norm{A'_2}_\infty$, and define
\begin{align*}
    \hat{H}_0 (\zeta) := D_x \sigma_1 + (\zeta - A_2(x)) \sigma_2 + m_0 \sigma_3.
\end{align*}
Then $\hat{H}_0 (\zeta)$ has a spectral gap in the interval $(-\Delta, \Delta)$, with $\Delta := \sqrt{m_0^2 - \norm{A'_2}_\infty}$.
\end{lemma}
The above implies that for $\hat{H}_0$, if $B_- < 0 < B_+$, every branch converging to a positive (resp. negative) value as $\zeta \rightarrow \infty$ goes to $+\infty$ (resp. $-\infty$) as $\zeta \rightarrow -\infty$.
\begin{proof}
Observe that $\hat{H}_0^2 (\zeta) := D^2_x + (\zeta - A_2(x))^2 + m_0^2 -A'_2 (x) \sigma_3$, and the result easily follows.
\end{proof}

We are now ready to complete the proof of Theorem \ref{prop:sf0}. The idea is to treat $\hat{H}$ as a perturbation of $\hat{H}_0$. 
For $\zeta \in \mathbb{R}$ and $(\lambda_1,\lambda_2) \in [0,1]^2$, define
\begin{align*}
    \hat{H} (\zeta; \lambda_1, \lambda_2) := D_x \sigma_1 + (\zeta - A_2(x)) \sigma_2 + m_0 \sigma_3 + \lambda_1 (m(x) - m_0) \sigma_3 + \lambda_2 V(x) \sigma_0,
\end{align*}
and let $\{\mu_j (\zeta;\lambda_1,\lambda_2)\}_{j \in \mathbb{Z}}$ denote the eigenvalues of $\hat{H} (\zeta;\lambda_1,\lambda_2)$.
We use the shorthand $\lambda:= (\lambda_1,\lambda_2)$ and $\mu_{j,\pm \infty}(\lambda) := \lim_{\zeta \rightarrow \pm \infty} \mu_j (\zeta;\lambda)$.
By Lemma \ref{lemma:mulambda},
the $\mu_j (\zeta;\lambda)$ are analytic in $(\zeta;\lambda)$
with 
$\partial_{\lambda_i} \mu_j (\zeta; \lambda)$ bounded uniformly in $(\lambda;\zeta)$ for $i \in \{1,2\}$.
This means
the limits as $\zeta \rightarrow \pm \infty$ of the $\mu_j (\zeta; \lambda)$ depend continuously on $\lambda$.
Recall that Lemmas \ref{lemma:mulim}, \ref{lemma:muInfty} and \ref{lemma:gap} give us a full description of the $\mu_j (\zeta; 0)$. 

When $B_-< 0 < B_+$, the 
$\mu_{j,\infty} (0)$ are known (and finite) and $\mu_{j,-\infty} (0) = \pm \infty$ if and only if $\pm \mu_{j,\infty} (0) > 0$.
The uniform bounds on $\partial_{\lambda_i} \mu_j$ imply that $\mu_{j,-\infty} (\lambda) = \mu_{j,-\infty} (0)$ for all $\lambda \in [0,1]^2$.
Combined with Lemma \ref{lemma:simple} and 
the fact that the multiset $\{\mu_{j,\infty} (1,1)\}_{j \in \mathbb{Z}}$ is known (Lemmas \ref{lemma:spectrumhat} and \ref{lemma:mulim}), this will allow us to obtain the limits as $\zeta \rightarrow \pm \infty$ of each $\mu_j (\zeta;1,1)$ via a smooth transition of $\lambda$ from $0$ to $(1,1)$. We find it easiest to first fix $\lambda_2 = 0$ while smoothly varying $\lambda_1$ from $0$ to $1$; then fix $\lambda_1 = 1$ while smoothly varying $\lambda_2$ from $0$ to $1$.
The case $B_+<0<B_-$ is handled similarly.

When $0<B_+,B_-$,
Lemma \ref{lemma:mulim} gives us the 
multisets $L_\pm := \{\mu_{j,\pm\infty} (\lambda)\}_{j\in \mathbb{Z}}$ 
for all $\lambda \in [0,1]^2$.
The only thing left is to pair each 
element of $L_-$ with an element of $L_+$ (i.e. for each $\mu_{j,-\infty} (\lambda) \in L_-$, determine $\mu_{j,+\infty} (\lambda)$).
When $\lambda = 0$, this pairing follows immediately from Lemma \ref{lemma:gap}.
That is (using a natural choice of indices), 
\begin{align*}
\dots,\quad \mu_{-1,\pm \infty} (0) = -\sqrt{2 |B_\pm| + m_0^2}, \quad \mu_{0,\pm \infty} (0) = m_0, \quad \mu_{1,\pm \infty} (0) = \sqrt{2 |B_\pm| + m_0^2},\quad \dots
\end{align*}
Since 
$\mu_{i,\pm \infty} (\lambda) \ne \mu_{j,\pm \infty} (\lambda)$
whenever $i \ne j$,
it is straightforward to obtain the limits of each $\mu_j$ for $\lambda = (1,1)$ (again by a smooth transition from $\lambda = 0$ to $\lambda = (1,1)$).
The case $B_+, B_- <0$ is handled similarly.

\begin{proof}[Proof of Theorem \ref{prop:sf0}]
Take $m_0$ as in Lemma \ref{lemma:gap}.
As in Lemma \ref{lemma:mulambda}
the eigenvalues and eigenprojections of $\hat{H} (\zeta; \lambda)$ are holomorphic in $(\zeta;\lambda)$.
It then follows from Lemma \ref{lemma:mulambda} that
\begin{align}\label{eq:mulambda}
    |\partial_{\lambda_1} \mu (\zeta, \lambda)| \le 
    \norm{m - m_0}_\infty, \qquad |\partial_{\lambda_2} \mu (\zeta, \lambda)| \le \norm{V}_\infty
\end{align}
for all eigenvalues $\mu (\zeta;\lambda)$ of $\hat{H} (\zeta;\lambda)$.
Define $$\hat{H}_{\pm} (\zeta;\lambda_1,\lambda_2) := D_x \sigma_1 + (\zeta - x B_\pm) \sigma_2 + m_0 \sigma_3 + \lambda_1 (m_\pm - m_0) \sigma_3 + \lambda_2 V_\pm \sigma_0.$$
Lemma \ref{lemma:spectrumhat} states that the spectrum of $\hat{H}_{\pm} (\zeta;0,0)$ consists entirely of eigenvalues and is given by
\begin{align*}
    \sigma (\hat{H}_{\pm} (\zeta;0,0)) &=
    \{\eps \sqrt{2k |B_\pm| + m_0^2} : \eps \in \{-1,1\}, k \in \mathbb{N}_+\} \bigcup \{m_0 \sgn (B_\pm)\}\\ &=: \{\tilde{\nu}_{k,\pm} : k \in \mathbb{Z}\}.
\end{align*}
Note that $|\tilde{\nu}_{k,\pm}| \ge |m_0| > 0$ for all $k$.

\begin{enumerate}
\item \label{it:infty} Suppose $B_- < 0 < B_+$. Recall that $\{\mu_j (\zeta; \lambda)\}_{j \in \mathbb{Z}}$ denotes the (holomorphic in $(\zeta; \lambda)$) eigenvalues of $\hat{H} (\zeta; \lambda)$. 
Then there is a bijection $\iota : \mathbb{Z} \rightarrow \mathbb{Z} \times \{+,-\}$ such that $\lim_{\zeta \rightarrow \infty} \mu_j (\zeta; 0) = \tilde{\nu}_{\iota (j)}$ for all $j \in \mathbb{Z}$.
Lemma \ref{lemma:muInfty} asserts that $|\mu_j (\zeta;0)| \rightarrow \infty$ as $\zeta \rightarrow -\infty$.
Hence Lemma \ref{lemma:gap} implies that $\mu_{j,-\infty} (0) = \pm \infty$ if and only if $\pm \tilde{\nu}_{\iota (j)} > 0$. That is, the levels $\tilde{\nu}_{\iota(j)} > 0$ correspond to branches $\mu_j (\cdot\; ;0)$ that go to $+\infty$ at $-\infty$, while the levels $\tilde{\nu}_{\iota (j)} < 0$ correspond to branches $\mu_j (\cdot \; ;0)$ that go to $-\infty$ at $-\infty$.
Thus $0$ (or any value in $(-|m_0|, |m_0|)$) separates the 
branches $\mu_j (\cdot \; ;0)$ that go to $+\infty$ at $-\infty$ from those that go to $-\infty$ at $-\infty$, as demonstrated by Figure \ref{fig:sf0} (bottom right panel). 

We now analyze the spectrum of $\hat{H} (\zeta;\lambda)$ as $\lambda$ is continuously deformed from $0$ to $(1,0)$ to $(1,1)$, thus separating the effects of $m$ and $V$.
Note that we can choose $m_0$ such that
the spectral flow of $H$ through $\alpha$ is well defined when $\lambda \in \{0,(1,0)\}$ (by definition the spectral flow is well defined when $\lambda = (1,1)$).
We 
follow the convention that $\mu_j (\zeta; \lambda) < \mu_{j+1} (\zeta;\lambda)$ and $\mu_{1,\infty} (0) = |m_0|$.
This means $\mu_{j,\infty} (0) > 0$ and $\mu_{j,-\infty} (0) = + \infty$ for all $j > 0$, while
$\mu_{j,\infty} (0) < 0$ and $\mu_{j,-\infty} (0) = - \infty$ for all $j \le 0$.
From \eqref{eq:mulambda}, it follows that for all $\lambda \in [0,1]^2$,
$\mu_{j,-\infty} (\lambda) = + \infty$ for all $j > 0$ and $\mu_{j,-\infty} (\lambda) = - \infty$ for all $j \le 0$.
Moreover, \eqref{eq:mulambda} implies that $\mu_{1,\infty} (1, 0) = \max \{m_+, m_-\}$ and $\mu_{0,\infty} (1, 0) = \min \{m_+, m_-\}$.
By Lemma \ref{lemma:mulim}, 
the values $\mu_{j,\infty} (1,0)$ can be read off directly from \eqref{eq:spectrumhat} with $V_\pm = 0$,
as they are exactly the elements of
$\sigma (\hat{H}_+ (0; 1,0))+ \sigma (\hat{H}_- (0;1,0)$ (in in increasing order with the limit of $\mu_1$ already determined above).
This confirms Theorem \ref{prop:sf0} in the case that $V \equiv 0$.

Now we set $\lambda_1 = 1$ and analyze the transition of $\lambda_2$ from $0$ to $1$.
Let $\{\nu_{0,i,\pm}\}_{i \in \mathbb{Z}}$ denote the eigenvalues of $\hat{H}_\pm (\zeta;1,0)$, where $\nu_{0,i,\pm} < \nu_{0,i+1,\pm}$ for all $i\in \mathbb{Z}$. 
The eigenvalues of $\hat{H}_\pm (\zeta; 1, \lambda_2)$ are then given by $\nu_{i,\pm}(\lambda_2) :=\nu_{0,i,\pm} + \lambda_2 V_\pm$ for $i \in \mathbb{Z}$.
To each $\nu_{i,\pm}(\lambda_2)$ is associated a unique index $j = j (\lambda_2,i,\pm)$
such that $\mu_{j(\lambda_2, i,\pm),\infty} (1,\lambda_2) = \nu_{i,\pm}(\lambda_2)$.
Let $\eps_{i, \pm}(\lambda_2) \in \{1,-1\}$ such that
$\eps_{i,\pm}(\lambda_2)=1$ if and only if $\mu_{j(\lambda_2,i,\pm), -\infty} (1,\lambda_2) = +\infty$. 
It follows that 
\begin{align}\label{eq:sfeps}
\begin{split}
\text{SF} &(H (1,\lambda_2); \alpha) =\\
&\sum_{\delta \in \{+,-\}} (|\{ i:\nu_{i,\delta}(\lambda_2) > \alpha, \; \eps_{i, \delta}(\lambda_2) = -1\}| -|\{i:\nu_{i,\delta}(\lambda_2) < \alpha, \; \eps_{i, \delta}(\lambda_2) = 1 \}|).
\end{split}
\end{align}
By \eqref{eq:mulambda}, we know that if $\nu_{i,+}(\lambda_2) \notin \{\nu_{i,-}(\lambda_2)\}_{i \in \mathbb{Z}}$ for all $\lambda_2$ in an open interval $(a,b)$, then $\eps_{i,+}(\lambda_2)$ is constant over $\lambda_2 \in (a,b)$.
For values $\lambda_2^*$ and indices $i,j$ such that
$\nu_{i,+}(\lambda_2^*) = \nu_{j,-}(\lambda_2^*)$, 
Lemma \ref{lemma:simple} implies that $\eps_{i,+}(\lambda_2)$ and $\eps_{j,-}(\lambda_2)$ trade signs across $\lambda_2^*$;
that is, $\eps_{i,+}(\lambda_2^* + \eta) = \eps_{j,-}(\lambda_2^* - \eta)$ and $\eps_{j,-}(\lambda_2^* + \eta) = \eps_{i,+}(\lambda_2^* - \eta)$ for all $\eta>0$ sufficiently small.
But 
this trade of signs has no effect on the spectral flow (i.e. if $\lim_{\zeta \rightarrow \infty} \tilde{\mu}_1 (\zeta) = 1$ and $\lim_{\zeta \rightarrow \infty} \tilde{\mu}_{0} (\zeta) = -1$ for some smooth branches of spectrum $\tilde{\mu}_1$ and $\tilde{\mu}_0$, then any spectral flow is independent of whether $\lim_{\zeta \rightarrow -\infty} (\tilde{\mu}_1 (\zeta), \tilde{\mu}_{0} (\zeta)) = (+\infty, -\infty)$ or $\lim_{\zeta \rightarrow -\infty} (\tilde{\mu}_1 (\zeta), \tilde{\mu}_{0} (\zeta)) = (-\infty, +\infty)$).
This means that when evaluating the right-hand side of \eqref{eq:sfeps}, we can replace $\eps_{i, \delta}(\lambda_2)$ by
$\eps_{i, \delta}(0)$ to obtain
\begin{align*}
    \text{SF} (H (1,1); \alpha) =\sum_{\delta \in \{+,-\}} (|\{ i\le 0 :\nu_{0,i,\delta} > \alpha-V_\delta\}| -|\{i>0:\nu_{0,i,\delta} < \alpha-V_\delta \}|),
\end{align*}
where we use the convention that $\eps_{i,\pm}(0) = 1$ if and only if $i > 0$.
To verify that the above
expression yields Theorem \ref{prop:sf0} is a straightforward but tedious exercise.
We will do so assuming that $m_- < 0 < m_+$, $|m_-| \le |m_+|$ and $\alpha - V_-, \alpha - V_+ > 0$, and leave the other cases (which are handled similarly) to the reader.

Under the above assumptions, it follows that
\begin{align}\label{eq:sfm}
    \text{SF} (H (1,1); \alpha) =M_0 -\sum_{\delta \in \{+,-\}} |\{i>0:\nu_{0,i,\delta} < \alpha-V_\delta \}|=: M_0-M_+ - M_-,
\end{align}
where 
\begin{align*}
    M_0=
    \begin{cases}
    1, & \alpha-V_- < |m_-|\\
    0, & \text{else},
    \end{cases}
\end{align*}
\begin{align*}
    M_+ =
    \begin{cases}
    0, & \alpha - V_+ < m_+\\
    k; & \sqrt{2 (k-1) B_+ + m_+^2} < \alpha-V_+ < \sqrt{2 k B_+ + m_+^2}, \quad k \in \mathbb{N}_+
    \end{cases}
\end{align*}
and 
\begin{align*}
    M_- =
    \begin{cases}
    0, & \alpha - V_- < \sqrt{2 |B_-| + m_-^2}\\
    k; & \sqrt{2 k |B_-| + m_-^2} < \alpha-V_- < \sqrt{2 (k+1) |B_-| + m_-^2}, \quad k \in \mathbb{N}_+.
    \end{cases}
\end{align*}
We will now show that \eqref{eq:sfm} agrees with Theorem \ref{prop:sf0}.
Simplifying the formula from Theorem \ref{prop:sf0}, we obtain that
\begin{align*}
    N(H_+;\alpha)=
    \begin{cases}
    0; & \alpha-V_+< \sqrt{2B_+ + m_+^2},\\
    k; & \sqrt{2 k B_+ + m_+^2} < \alpha-V_+ < \sqrt{2 (k+1) B_+ + m_+^2}, \quad k \in \mathbb{N}_+
    \end{cases}
\end{align*}
and
\begin{align*}
    N(H_-;\alpha)=
    \begin{cases}
    0; & \alpha-V_-< \sqrt{2|B_-| + m_-^2},\\
    k; & \sqrt{2 k |B_-| + m_-^2} < \alpha-V_- < \sqrt{2 (k+1) |B_-| + m_-^2}, \quad k \in \mathbb{N}_+.
    \end{cases}
\end{align*}
Note that
\begin{align*}
    I(H_+;\alpha) &= 
\sgn (\alpha-V_+-m_+) (N(H_+;\alpha) + \frac{1}{2}),\\
I(H_-;\alpha) &= -\sgn (\alpha-V_- +m_-)(N(H_-;\alpha) + \frac{1}{2}).
\end{align*}
Observe that $N(H_\pm;\alpha) = 0$ if $\sgn (\alpha - V_\pm \mp m_\pm) < 0$, hence
\begin{align*}
I(H_+;\alpha) = 
    \begin{cases}
    N(H_+;\alpha) + \frac{1}{2}, & \alpha - V_+ - m_+ > 0\\
    -\frac{1}{2}, & \alpha - V_+ - m_+ < 0
    \end{cases}
\end{align*}
and
\begin{align*}
    I(H_-;\alpha) = 
    \begin{cases}
    -N(H_-;\alpha) - \frac{1}{2}, & \alpha - V_- + m_- > 0\\
    \frac{1}{2}, & \alpha - V_- + m_- < 0.
    \end{cases}
\end{align*}
Observe that $N(H_-;\alpha) = M_-$ and
\begin{align*}
    N(H_+;\alpha) =
    \begin{cases}
    M_+, & \alpha - V_+ < m_+\\
    M_+ -1, & \text{else}.
    \end{cases}
\end{align*}
Thus
\begin{align*}
I(H_+;\alpha) = 
    \begin{cases}
    M_+ - \frac{1}{2}, & \alpha - V_+ - m_+ > 0\\
    -\frac{1}{2}, & \alpha - V_+ - m_+ < 0
    \end{cases}
\end{align*}
and
\begin{align*}
    I(H_-;\alpha) = 
    \begin{cases}
    -M_- - \frac{1}{2}, & \alpha - V_- + m_- > 0\\
    \frac{1}{2}, & \alpha - V_- + m_- < 0.
    \end{cases}
\end{align*}
Since $M_\pm = 0$ whenever $\alpha - V_\pm -\mp m_\pm < 0$,
we conclude that
\begin{align*}
    I(H_-;\alpha) - I(H_+;\alpha) = \begin{cases}
    -M_- - M_+, & \alpha - V_- + m_- > 0,\\
    -M_- - M_+ + 1, & \alpha - V_- + m_- < 0.
    \end{cases}
\end{align*}
This agrees with \eqref{eq:sfm}, as desired.


\item \label{it:finite} Suppose $0 < B_+, B_-$. The argument is similar to case \ref{it:infty},
only now each $\mu (\zeta; \lambda)$ converges also as $\zeta \rightarrow -\infty$. 
By Lemmas \ref{lemma:simple} and \ref{lemma:gap}, we know that
the branch $\mu (\zeta; 0)$ that converges to $m_0$ as $\zeta \rightarrow -\infty$ also converges to $m_0$ as $\zeta \rightarrow +\infty$.
Thus \eqref{eq:mulambda} implies that
the branch $\mu (\zeta; 1,0)$ that converges to $m_-$ as $\zeta \rightarrow -\infty$ must converge to $m_+$ as $\zeta \rightarrow +\infty$. Finally, this implies that the branch $\mu (\zeta; 1,1)$ that converges to $m_- + V_-$ as $\zeta \rightarrow -\infty$ must converge to $m_+ + V_+$ as $\zeta \rightarrow \infty$, and the result follows.
\end{enumerate}
The case $B_+ < 0 < B_-$ is handled similarly to case \ref{it:infty}; the case $B_+, B_- < 0$ is handled similarly to case \ref{it:finite}. This completes the result.
\end{proof}

\section{Physical observable}\label{sec:stability}
Let $P (y) = P \in \fs (0,1)$ and $\varphi \in \fs (0,1; E_1, E_2)$ for some $[E_1, E_2] \subset \rho (H_+) \cap \rho (H_-)$.
Define the interface conductivity associated to $H$ in \eqref{eq:md} by 
\begin{equation}\label{eq:sigmaI}
 \sigma_I := \Tr i [H,P] \varphi ' (H).
\end{equation}
The goal of this section is to relate $\sigma_I$ to the spectral flow from Section \ref{sec:spec}, and to prove its stability with respect to perturbations of $H$.
To do so, we will use well-known results on pseudo-differential operators ($\Psi$DOs) and the Helffer-Sj\"ostrand formula, which are summarized in Appendix \ref{sec:PsiDO} 
along with the notation that will be used below. 
Let $$\sigma (x,\xi,\zeta) := \xi \sigma_1 + (\zeta - A_2 (x)) \sigma_2 + m(x) \sigma_3 + V(x) \sigma_0$$ denote the \emph{Weyl symbol} of $H$ (so that $H = \Op (\sigma)$).
Similarly, define
\begin{align*}
    \sigma_\pm (x,\xi,\zeta) := \xi \sigma_1 + (\zeta - x B_\pm) \sigma_2 + m_\pm \sigma_3 + V_\pm \sigma_0
\end{align*}
so that $H_\pm = \Op (\sigma_\pm)$.
Note that as opposed to the setting considered in \cite{bal2022mathematical,QB}, the symbols $\sigma_\pm$ still depend on the spatial variable $x$.

As mentioned in Section \ref{sec:intro}, $\sigma$ is not elliptic because its eigenvalues can stay bounded even as $x$ and $\zeta$ get large. Still, using the fact that
\begin{align*}
    \sigma^2 (x,\xi,\zeta) = \xi^2 + (\zeta - A_2 (x))^2 + m^2(x) + V^2 (x) + 2 V(x) (\xi \sigma_1 + (\zeta - x B_\pm) \sigma_2 + m_\pm \sigma_3),
\end{align*} 
we obtain that
\begin{align}\label{eq:ell}
    |\det \sigma (x,\xi,\zeta)|^{1/2} \ge c_1 \aver{\xi, \zeta - A_2 (x)} - c_2
\end{align}
for some $0 < c_1 < 1$ and $c_2 >0$. In order to make use of this inequality, we need to verify
\begin{lemma} \label{lemma:of}
$\aver{\xi, \zeta - A_2 (x)}$ is an order function.
\end{lemma}
We refer to Appendix \ref{sec:PsiDO} for the definitions of an order function and $\aver{\cdot}$.
\begin{proof}
It is known (see e.g. \cite{Zworski}) that $\mathbb{R}^2 \ni Y \mapsto\aver{Y}$ is an order function. Thus there exist positive constants $C_0$ and $N_0$ such that
\begin{align*}
    \aver{\xi_1, \zeta_1 - A_2 (x_1)} \le C_0 \aver{\xi_2 - \xi_1, \zeta_2 - A_2 (x_2) - (\zeta_1 - A_2 (x_1))}^{N_0} \aver{\xi_2, \zeta_2 - A_2 (x_2)}
\end{align*}
for all $(x_1,\xi_1,\zeta_1), (x_2,\xi_2,\zeta_2) \in \mathbb{R}^3$.
We write $$(\zeta_2 - A_2 (x_2) - (\zeta_1 - A_2 (x_1)))^2 \le 2 ((\zeta_2 - \zeta_1)^2 + (A_2 (x_2) - A_2 (x_1))^2,$$ 
and seek to bound the second term on the above right-hand side. We have
\begin{align*}
    (A_2 (x_2) - A_2 (x_1))^2 = (x_2 B (x_2) - x_1 B (x_1))^2 &= ((x_2 - x_1) B (x_2) + x_1 (B(x_2) - B(x_1)))^2\\
    &\le 2 \norm{B}_\infty^2 (x_2 - x_1)^2 + 2x_1^2 (B (x_2)- B(x_1))^2.
\end{align*}
Since $B(x_2) - B(x_1)$ vanishes whenever $x_1$ and $x_2$ are sufficiently large and of the same sign, there exist positive constants $C_1$ and $C_2$ such that
$x_1^2 (B (x_2)- B(x_1))^2 \le C_1 (x_2-x_1)^2 + C_2$.
We conclude that
\begin{align*}
    \aver{\xi_2 - \xi_1, \zeta_2 - A_2 (x_2) - (\zeta_1 - A_2 (x_1))} \le C \aver{x_2 - x_1, \xi_2 - \xi_1, \zeta_2 - \zeta_2}
\end{align*}
for some $C>0$, and the result is complete.
\end{proof}

We begin by showing that $\sigma_I$ is well defined.
To do this, we will need
two useful decay properties for symbols of related operators.
\begin{lemma}\label{lemma:resolvent}
If $z \in \mathbb{C}$ such that $\Im z \ne 0$,
then $(z-H)^{-1} = \Op (r_{z})$ for some $r_{z} \in S(\aver{\xi, \zeta - A_2 (x)}^{-1})$.
\end{lemma}
Note that $S(\aver{\xi, \zeta - A_2 (x)}^{-1})$ is well defined, by Lemma \ref{lemma:of}.
\begin{proof}
Since $H$ is self-adjoint, we know that $(z-H)^{-1}$ is well defined and bounded.
To obtain bounds for the symbol of $(z-H)^{-1}$ (and show that it is a $\Psi$DO in the first place), we use Beals's criterion presented in \cite[Proposition 8.3]{DS}. This result states that $(z-H)^{-1} = \Op (r_z)$ for some $r_z \in S(1)$ if and only if for any collection of linear forms $$\ell_1 (x,y,\xi,\zeta), \ell_2 (x,y,\xi,\zeta), \dots , \ell_N (x,y,\xi,\zeta)$$ on $\mathbb{R}^4$, the operator $\ad_{L_1} \circ \dots \circ \ad_{L_N} \circ (z-H)^{-1}$ is bounded in $L^2 (\mathbb{R}^2) \otimes \mathbb{C}^2$, where $L_j := \Op (\ell_j)$ and $\ad_A B := [A,B]$.
Since $\sigma_\xi = \sigma_1$ and $\sigma_\zeta =\sigma_2$ are constant and $\sigma_x$ is bounded,
it is clear that $[L_j,H]$ is bounded for any such $L_j$.
Thus the identity $[\mathcal{O}, (z-H)^{-1}] = (z-H)^{-1} [\mathcal{O}, H] (z-H)^{-1}$ easily implies that
$(z-H)^{-1} = \Op (r_{z})$ for some $r_{z} \in S(1)$.

By \eqref{eq:ell} and the composition calculus, 
we have
\begin{align*}
    (z-\sigma) \sharp (z-\sigma)^{-1} = 1 + b_{z},
\end{align*}
where
$b_{z} \in S (\aver{\xi, \zeta - A_2 (x)}^{-2})$.
Indeed, all derivatives of $\sigma$ are bounded and $(z-\sigma)^{-1} \in S(\aver{\xi,\zeta - A_2(x)}^{-1})$.
Letting $G_{z} := \Op ((z-\sigma)^{-1})$ and $B_{z} := \Op (b_{z})$, this means
\begin{align*}
    (z-H) G_{z} = 1 + B_{z}.
\end{align*}
Applying $(z-H)^{-1}$ to both sides (on the left), we get
\begin{align*}
    (z-H)^{-1} = G_{z} - (z-H)^{-1} B_{z}.
\end{align*}
The first term on the right-hand side has symbol in $S(\aver{\xi,\zeta - A_2(x)}^{-1})$ and the second term has symbol in $S(\aver{\xi,\zeta - A_2(x)}^{-2})\subset S(\aver{\xi,\zeta - A_2(x)}^{-1})$.
Therefore,
$r_{z} \in S(\aver{\xi, \zeta - A_2 (x)}^{-1})$ as desired.
\end{proof}

\begin{lemma} \label{lemma:decayPhiH}
For any 
$\Phi \in \mathcal{C}^\infty_c (E_1, E_2)$, we have $\Phi (H) \in \Op (S(\aver{x,\xi,\zeta}^{-\infty}))$.
\end{lemma}

\begin{proof}
For any $p>0$, we can write $\Phi (H) = (i-H)^{-p} \Phi_p (H)$ with $\Phi_p \in \mathcal{C}^\infty_c (E_1, E_2)$.
By Lemma \ref{lemma:resolvent} and the composition calculus, this means
$\Phi (H) \in \Op (S(\aver{\xi,\zeta-A_2(x)}^{-\infty}))$.
Since $H_\pm$ has a spectral gap in $[E_1, E_2]$, we know that $\Phi (H_\pm) = 0$.
Thus we can write $\Phi (H) = \phi(x) (\Phi (H) - \Phi (H_+)) + (1-\phi (x)) (\Phi (H) - \Phi (H_-))$, for some $\phi \in \fs (0,1)$.
The Helffer-Sj\"ostrand formula implies that
\begin{align*}
    \Phi (H) - \Phi (H_+) = \frac{1}{\pi}\int_{\mathbb{C}} \bar{\partial} \tilde{\Phi} (z) (z-H)^{-1} (H - H_+) (z-H_+)^{-1} d^2 z.
\end{align*}
Since $\sigma - \sigma_+$ vanishes whenever $x$ is sufficiently large, it follows that $\Phi (H) - \Phi (H_+) \in \Op (S (\aver{x_+}^{-\infty}))$, where $x_+ := \max \{x, 0\}$. Since $\phi$ vanishes whenever $-x$ is sufficiently large, we conclude that $\phi(x) (\Phi (H) - \Phi (H_+)) \in \Op (S (\aver{x}^{-\infty}))$. The same reasoning shows that $(1-\phi (x)) (\Phi (H) - \Phi (H_-)) \in \Op (S (\aver{x}^{-\infty}))$.
We have thus shown that
$\Phi (H) \in \Op (S(\aver{\xi,\zeta-A_2(x)}^{-\infty}) \cap S (\aver{x}^{-\infty}))$. By interpolation, the result is complete.
\end{proof}
\begin{lemma}
For any $\Phi \in \mathcal{C}^\infty_c (E_1, E_2)$, the operator $[H,P] \Phi (H)$ is trace-class.
\end{lemma}
\begin{proof}
We have $[H,P] = -i P'(y) \sigma_2$ with $P' \in \mathcal{C}^\infty_c$, hence $[H,P] \in \Op (S (\aver{y}^{-\infty}))$. The result then follows from Lemma \ref{lemma:decayPhiH} and the composition calculus.
\end{proof}
Now that we have shown that $\sigma_I$ is well defined, we relate it to the spectral flow.
\begin{theorem}\label{prop:sf}
For any $\alpha \in [E_1, E_2]$, we have $2 \pi \sigma_I = \text{SF} (\alpha)$.
\end{theorem}
Combining Theorems \ref{prop:sf0} and \ref{prop:sf}, we obtain an explicit formula for $\sigma_I$. In particular, $\sigma_I$ is quantized and independent of compact perturbations in $m, B,$ and $V$. Moreover, $\sigma_I$ is stable with respect to sufficiently small changes in $m_\pm, B_\pm$ and $V_\pm$.
\begin{proof}
We follow the arguments presented in \cite[Section A]{3}.
By Lemma \ref{lemma:decayPhiH} and \cite[Lemma 3.4]{QB}, we obtain that
$\sigma_I = \Tr i [\Psi(H),P] \varphi ' (H)$ for any $\Psi \in \mathcal{C}^\infty_c (E_1, E_2)$ that satisfies $\Psi (\lambda) = \lambda$ for all $\lambda$ in some open interval containing $\supp (\varphi ')$.
The kernel of $[\Psi(H),P] \varphi'(H)$ is 
\begin{align*}
    t(x, x'; y, y') = 
    \int _{\mathbb{R}^{2}} (P(y'') - P(y)) k_\Psi(x,x''; y-y'') k_{\varphi '}(x'',x'; y''-y') dx'' dy'',
\end{align*}
where $k_\Phi$ for $\Phi \in \mathcal{C}^\infty_c (E_1, E_2)$ is given by \eqref{eq:kernelPhi}.
It follows that
\begin{align*}
    \sigma_I &= i \tr \int_{\mathbb{R}^2} t(x,x;y,y) dx dy \\&=
    i \tr \int_{\mathbb{R}^4}(P(y') - P(y)) k_\Psi(x,x'; y-y') k_{\varphi '}(x',x; y'-y) dx' dy' dx dy.
\end{align*}
Changing integration variables $(y, y') \rightarrow (z, y')$ with $z = y-y'$, and using that $\int_\mathbb{R} P (y') - P(y'+z) dy' = -z$ (which follows from $P \in \fs (0,1)$), we obtain
\begin{align*}
    \sigma_I = -i \tr \int_{\mathbb{R}^3} z k_\Psi(x,x'; z) k_{\varphi '}(x',x; -z) dx'dx dz.
\end{align*}
By Parseval and using that $k_{\varphi '}^* (x,x'; z) = k_{\varphi '} (x',x;-z)$, we have
\begin{align*}
    \sigma_I = \frac{1}{2\pi} \tr \int_{\mathbb{R}^3} \partial_\zeta \hat{k}_\Psi (x,x';\zeta) \hat{k}_{\varphi '}^* (x,x',\zeta) dx' dx d\zeta.
\end{align*}
Note that for any $\Phi \in \mathcal{C}^\infty_c (E_1, E_2)$, 
\begin{align*}
    \hat{k}_\Phi (x,x';\zeta) =  \sum_{j=1}^{J} \Phi(\mu_j(\zeta)) \psi_j(x, \zeta) \psi_j^*(x', \zeta),
\end{align*}
hence
\begin{align*}
    2\pi\sigma_I &= \tr \int_{\mathbb{R}^3} \sum_{j,k = 1}^J \partial_\zeta (\Psi(\mu_j(\zeta)) \psi_j(x, \zeta) \psi_j^*(x', \zeta))
    \varphi '(\mu_k(\zeta)) \psi_k(x', \zeta) \psi_k^*(x, \zeta) dx' dx d\zeta\\
    &= \tr \int_{\mathbb{R}^3} \sum_{j,k = 1}^J \partial_\zeta (\mu_j(\zeta) \psi_j(x, \zeta) \psi_j^*(x', \zeta))
    \varphi '(\mu_k(\zeta)) \psi_k(x', \zeta) \psi_k^*(x, \zeta) dx' dx d\zeta.
\end{align*}
One can verify that the contribution from $\partial_\zeta (\psi_j(x, \zeta) \psi_j^*(x', \zeta))$ vanishes, and thus
\begin{align*}
    2\pi \sigma_I &=\tr \int_{\mathbb{R}^3} \sum_{j,k = 1}^J \partial_\zeta \mu_j(\zeta) \psi_j(x, \zeta) \psi_j^*(x', \zeta)
    \varphi '(\mu_k(\zeta)) \psi_k(x', \zeta) \psi_k^*(x, \zeta) dx' dx d\zeta\\
    &=
    \int_{\mathbb{R}} \sum_{j = 1}^J \partial_\zeta \mu_j(\zeta)
    \varphi '(\mu_j(\zeta)) d\zeta = \sum_{j = 1}^J \int_{\mathbb{R}} \partial_\zeta
    \varphi(\mu_j(\zeta)) d\zeta,
\end{align*}
where we have used orthonormality of the eigenfunctions to justify the second equality.
Hence
\begin{align*}
    2\pi \sigma_I = \sum_{j=1}^J \Big( \lim_{\zeta \rightarrow +\infty} \varphi (\mu_j (\zeta)) - \lim_{\zeta \rightarrow -\infty} \varphi (\mu_j (\zeta)) \Big).
\end{align*}
Since $\varphi \in \fs (0,1;E_1, E_2)$, the above right-hand side is indeed well defined and the result is complete.
\end{proof}

Now we want to analyze the stability of $\sigma_I$ with respect to perturbations. Let 
\begin{equation}\label{eq:Hmu}
    H_\mu = H + \mu W \quad \mbox{ for } \quad\mu \in [0,1],
\end{equation}
where $W$ is a symmetric $\Psi$DO with symbol in $S (\aver{x,y,\xi,\zeta}^s)$ for some $s \in \mathbb{R}$. 
We begin with a criterion for stability of $\sigma_I$. It will require two preliminary results (Lemmas \ref{lemma:pinv} and \ref{lemma:xinfinity}), which can also be found in \cite[Section 3]{QB}.
Below, $U^\circ$ denotes the interior of $U$.
\begin{theorem}\label{thm:sc}
Let $\Psi \in \mathcal{C}^\infty_c (E_1,E_2)$ such that $\varphi ' \in \mathcal{C}^\infty_c (\{\Psi (\lambda) = \lambda\}^\circ)$.
If $H_\mu$ is self-adjoint and the operators $[H_\mu, P] \varphi ' (H_\mu), \; \varphi ' (H_\mu) - \varphi '(H_0)$ and $\Psi (H_\mu) - \Psi (H_0)$ are trace-class, then $\sigma_I (H_\mu) = \sigma_I (H_0)$. 
\end{theorem}
\begin{proof}
By assumption, $\sigma_I (H_\mu)$ is well defined.
It follows from cyclicity of the trace \cite{Kalton} 
that $\sigma_I (H_\lambda) = \Tr i [\Psi (H_\lambda), P] \varphi ' (H_\lambda)$ for $\lambda \in \{0,\mu\}$; see \cite[Lemma 3.4]{QB} for more details.
We can thus write the difference of conductivities as
\begin{align*}
    \sigma_I (H_\mu) - \sigma_I (H_0) = \Tr i [\Psi (H_\mu),P] (\varphi ' (H_\mu) - \varphi ' (H_0)) + \Tr i [\Psi (H_\mu) - \Psi (H_0), P] \varphi ' (H_0),
\end{align*}
where our hypotheses have guaranteed that each trace is well defined. Using Lemma \ref{lemma:pinv}, we can replace $P$ above by $P_{y_0}$, where $P_{y_0} (y) := P (y-y_0)$. Again applying cyclicity (and linearity) of the trace, we get
$
    \sigma_I (H_\mu) - \sigma_I (H_0) = \sum_{j=1}^4 \Tr P_{y_0} A_j,
$
where the $A_j$ are all trace-class. The result then follows from Lemma \ref{lemma:xinfinity}.
\end{proof}

\begin{lemma} \label{lemma:pinv}
Let $P_1, P_2 \in \fs (0,1)$ with $P_j = P_j (y)$.
Then 
\begin{align*}
    \Tr i [H,P_1] \varphi '(H) = \Tr i [H,P_2] \varphi '(H).
\end{align*}
\end{lemma}

\begin{proof}
With $\Psi$ defined in Theorem \ref{thm:sc}, we have
\begin{align*}
    \Tr i [H,P_2] \varphi '(H) -\Tr i [H,P_1] \varphi '(H) = 
    \Tr i [\Psi (H), P_2 - P_1] \varphi '(H).
\end{align*}
Since $P_2 - P_1 \in \aver{x}^{-\infty}$,
Lemma \ref{lemma:decayPhiH} implies that $(P_2 - P_1) \varphi ' (H)$ is trace-class.
Therefore,
\begin{align*}
    \Tr i [\Psi (H), P_2 - P_1] \varphi '(H) &= 
    \Tr i \Psi (H) (P_2 - P_1) \varphi '(H) - \Tr i (P_2 - P_1) \Psi (H) \varphi '(H)\\
    &=
    \Tr i (P_2 - P_1) \varphi ' (H) \Psi(H) -
    \Tr i (P_2 - P_1) \Psi (H) \varphi '(H) =0,
\end{align*}
where we have used cyclicity of the trace to justify the second line, and the fact that $[\varphi ' (H), \Psi (H)] = 0$ for the last line.
\end{proof}

\begin{lemma}\label{lemma:xinfinity}
Let $P(y) = P \in \fs{(0,1)}$ and $y_0 \in \mathbb{R}$, and define $P_{y_0} (y) := P(y-y_0)$. Then for any trace-class operator $A$ on $L^2 (\mathbb{R}^2) \otimes \mathbb{C}^n$, we have $\Tr P_{y_0} A \rightarrow 0$ as $y_0 \rightarrow \infty$.
\end{lemma}

\begin{proof}
Writing $A = \frac{1}{2}(A+A^*) + \frac{1}{2}(A-A^*)$, with $\frac{1}{2}(A+A^*)$ and $\frac{i}{2}(A-A^*)$ trace-class and self-adjoint and using the triangle inequality, we may assume that $A$ is self-adjoint.
Fix $\eps > 0$.
By the spectral theorem, 
there exists an orthonormal basis $\{\psi_j \}_{j=1}^\infty$ of $L^2 (\mathbb{R}^2) \otimes \mathbb{C}^n$ such that $A \psi_j = \lambda_j \psi_j$, with $\{\lambda_j\}_{j=1}^\infty \subset \mathbb{R}$ satisfying 
$\sum_{j=1}^\infty |\lambda_j| < \infty$.
Thus there exists $N \in \mathbb{N}$ such that
$
    \sum_{j=N}^\infty |(\psi_j, P_{y_0} A \psi_j)| 
    \le \sum_{j=N}^\infty |\lambda_j| < \eps/2
$
for all $y_0 \in \mathbb{R}$.
Since 
$\{\psi_j \}_{j=1}^\infty \subset L^2 (\mathbb{R}^2) \otimes \mathbb{C}^n$, 
we know that,
for $y_0$ sufficiently large,
$\norm{P_{y_0} \psi_j} < \eps/\left(2N\norm{A}\right)$ for all $j \in \{1,2, \dots, N-1\}$.
It follows that
$|\Tr P_{y_0} A| \le \sum_{j=1}^\infty |(\psi_j, P_{y_0} A \psi_j)| < \eps$ for all $y_0$ sufficiently large.
\end{proof}

We now use Theorem \ref{thm:sc} to prove stability of $\sigma_I$ under a large class of perturbations $W$.
For the rest of this section, let
$\Phi_0 \in \mathcal{C}^\infty_c (E_1, E_2)$ such that $\Psi \in \mathcal{C}^\infty_c (\{ \Phi_0 = 1 \}^\circ)$, with $\Psi$ as in Theorem \ref{thm:sc}.

\begin{theorem}\label{prop:stabilitysmall}
Let $W$ be a symmetric $\Psi$DO with symbol in $S (\aver{x,y,\xi,\zeta})$. Assume that $(i \pm H_0)^{-1} W$ is bounded and $\Phi_0 (H_0) W$ is trace-class. 
Then 
\begin{align}\label{eq:musmall}
\sigma_I (H_\mu) = \sigma_I (H_0) \qquad \text{for all} \quad \mu < \min \Big \{\norm{(i+ H_0)^{-1} W}^{-1},\norm{(i- H_0)^{-1} W}^{-1}\Big \}.
\end{align}
\end{theorem}

\begin{proof}
We verify the assumptions of Theorem \ref{thm:sc}.
Since $(i\pm H_0)^{-1} W$ is bounded, 
$i \pm H_\mu : (i-H_0)^{-1} \mathcal{H} \rightarrow \mathcal{H}$ is bijective
whenever $\mu < \norm{(i\pm H_0)^{-1} W}^{-1}$,
with $(i\pm H_\mu)^{-1} = (1 \pm \mu (i \pm H_0)^{-1} W)^{-1} (i \pm H_0)^{-1}$.
Hence for all $\mu$ satisfying \eqref{eq:musmall}, $H_\mu$ is self-adjoint with the same domain of definition $\mathcal{D} (H_\mu) = \mathcal{D} (H_0) = (i-H_0)^{-1} \mathcal{H}$.

It remains to verify that $[H_\mu, P] \varphi ' (H_\mu)$ and $\Phi (H_\mu) - \Phi (H_0)$ are trace-class for $\Phi \in \{\varphi ' , \Psi\}$.
We start with the difference \begin{align*}\Phi (H_\mu) - \Phi (H_0) = (\Phi (H_\mu) &- \Phi (H_0))(\Phi_0 (H_\mu)- \Phi_0 (H_0))\\ &+ \Phi (H_0)(\Phi_0 (H_\mu) - \Phi_0 (H_0)) + (\Phi (H_\mu) - \Phi (H_0))\Phi_0 (H_0),\end{align*}
which can be rearranged to give
\begin{align}\label{eq:diff}
\begin{split}
    (\Phi (H_\mu) - \Phi (H_0))(1 &- (\Phi_0 (H_\mu)- \Phi_0 (H_0)))\\ &= \Phi (H_0)(\Phi_0 (H_\mu) - \Phi_0 (H_0)) + (\Phi (H_\mu) - \Phi (H_0))\Phi_0 (H_0).
    \end{split}
\end{align}
By the Helffer-Sj\"ostrand formula,
\begin{align}\label{eq:HS}
    \Phi_0 (H_\mu)- \Phi_0 (H_0) = \frac{1}{\pi} \int_Z \bar{\partial} \tilde{\Phi}_0 (z) (z-H_0)^{-1} \mu W (z-H_\mu)^{-1} d^2 z.
\end{align}
Hence the norm of $\Phi_0 (H_\mu)- \Phi_0 (H_0)$ can be bounded by $C \mu$, meaning that it is less than $1$ for $\mu$ small enough. It thus suffices to show that each term on the right-hand side of \eqref{eq:diff} is trace-class, as $1 - (\Phi_0 (H_\mu)- \Phi_0 (H_0))$ has bounded inverse.
Using \eqref{eq:HS}, the first term becomes
\begin{align}\label{eq:rhs1}
    \Phi (H_0)(\Phi_0 (H_\mu) - \Phi_0 (H_0)) =
    \frac{1}{\pi} \int_Z \bar{\partial} \tilde{\Phi}_0 (z) (z-H_0)^{-1} \mu\Phi (H_0) W (z-H_\mu)^{-1} d^2 z.
\end{align}
Since each resolvent has norm bounded by $C|\Im z|^{-1}$, with $|\bar{\partial} \tilde{\Phi}_0 (z)| \le C |\Im z|^2$ and $\Phi (H_0)W$ trace-class, we conclude that \eqref{eq:rhs1} is trace-class.
Applying the same argument to the second term 
\begin{align}\label{eq:rhs2}
    (\Phi (H_\mu) - \Phi (H_0))\Phi_0 (H_0) =-\frac{1}{\pi} \int_Z \bar{\partial} \tilde{\Phi} (z) (z-H_\mu)^{-1} \mu W \Phi_0 (H_0)(z-H_0)^{-1} d^2 z,
\end{align}
it follows that $\Phi (H_\mu) - \Phi (H_0)$ is trace-class.

For the first operator, we write
\begin{align*}
    [H_\mu, P] \Phi (H_\mu) = [H_\mu, P] \Phi (H_0) + [H_\mu, P] (\Phi (H_\mu)- \Phi (H_0)),
\end{align*}
with the first term on the above right-hand side trace-class by Lemma \ref{lemma:decayPhiH} and the composition calculus.
To show that the second term is also trace-class, we again use \eqref{eq:diff}.
Namely, it suffices to show that
\begin{align*}
    T_1 := [H_0, P] A, \qquad T_2 := [W,P]A, \qquad T_3:= [H_0, P] B, \qquad T_4 := [W,P] B
\end{align*}
are trace-class, with $A := \Phi (H_0)(\Phi_0 (H_\mu) - \Phi_0 (H_0))$ and $B := (\Phi (H_\mu) - \Phi (H_0))\Phi_0 (H_0)$.
That $T_1$ is trace-class follows immediately from \eqref{eq:rhs1} and boundedness of $[H_0, P] = -i P' (y) \sigma_2$. 
We know that $PWA$ is trace-class (since $(i \pm H_0)^{-1} W$ is bounded; this is an assumption on $W$), thus $T_2$ is trace-class if $WPA$ is.
We write
\begin{align*}
    WPA&=
    \frac{1}{\pi} \int_Z \bar{\partial} \tilde{\Phi}_0 (z) WP (z-H_0)^{-1} \mu\Phi (H_0) W (z-H_\mu)^{-1} d^2 z\\
    &=\frac{1}{\pi} \int_Z \bar{\partial} \tilde{\Phi}_0 (z) W(z-H_0)^{-1} P \mu\Phi (H_0) W (z-H_\mu)^{-1} d^2 z\\
    &\qquad +
    \frac{1}{\pi} \int_Z \bar{\partial} \tilde{\Phi}_0 (z) W[P, (z-H_0)^{-1}] \mu\Phi (H_0) W (z-H_\mu)^{-1} d^2 z.
\end{align*}
Since $[P, (z-H_0)^{-1}] = -(z-H_0)^{-1} [H_0, P] (z-H_0)^{-1}$, boundedness of $[H_0, P]$ and our assumption on $W$ imply that $WPA$ is trace-class.
Using \eqref{eq:rhs2} and the identity \begin{align*}(z-H_\mu)^{-1} &= (i-H_\mu)^{-1} (1 + (i-z)(z-H_\mu)^{-1})\\ &=(i-H_0)^{-1} (1-\mu W(i-H_0)^{-1})^{-1}(1 + (i-z)(z-H_\mu)^{-1}),\end{align*}
we similarly conclude that $T_3$ and $T_4$ are trace-class. This completes the result.
\end{proof}


We now 
prove a stability result that does not require the perturbation to be ``small''. To do this, 
we will need the stronger assumption that
$(i \pm H_0)^{-1} W$ is compact. 
\begin{theorem}\label{prop:stabilitycompact}
Let $W$ be a symmetric $\Psi$DO with symbol in $S (\aver{x,y,\xi,\zeta})$. 
Assume that $(i \pm H_0)^{-1} W$ is compact and $\Phi_0 (H_0) W$ is trace-class. 
Then $\sigma_I (H_1) = \sigma_I (H_0)$.
\end{theorem}
\begin{proof}
We again verify the hypotheses of Theorem \ref{thm:sc}.
The fact that
$(i \pm H_0)^{-1} W$ is compact implies that $H_1 = H_0+W$ is self-adjoint with domain of definition $\mathcal{D} (H_1) = \mathcal{D} (H_0)$.
Indeed, the fact that $W$ is symmetric implies that the kernel of $i + H_1$ is trivial. The Fredholm alternative and our compactness assumption imply that the dimension of the kernel of $1 + (i+H_0)^{-1} W$ is equal to the codimension of its range.
But $1 + (i+H_0)^{-1} W = (i+H_0)^{-1} (i+H_1)$ with $(i+H_0)^{-1} : \mathcal{H} \rightarrow \mathcal{D} (H_0)$ a bijection, and thus $\codim \Ran (i+H_1) = \dim \ker (i+H_1) = 0$.
The same argument also shows that
$\codim \Ran (i-H_1) = \dim \ker (i-H_1) = 0$.
We conclude by \cite[Theorem VIII.3]{RS} that $H_1$ is self-adjoint.

We now prove the necessary trace-class properties.
Let $\Phi \in \{\varphi ' ,\Psi \}$.
By \eqref{eq:diff}, \eqref{eq:rhs1} and \eqref{eq:rhs2}, we have that $\Theta (1-\Theta_0)$ is trace-class, where $\Theta:= \Phi (H_1) - \Phi (H_0)$ and $\Theta_0:= \Phi_0 (H_1) - \Phi_0 (H_0)$. We know that $\Theta_0$ is compact by \eqref{eq:HS} and our assumption that $W$ is relatively compact with respect to $H_0$.
Thus there exist $\Theta_{00}$ and $\Theta_{01}$ such that $\Theta_{00}$ has finite rank, $\norm{\Theta_{01}} < 1$, and $\Theta_0 = \Theta_{00} + \Theta_{01}$. Since $\Theta$ is bounded, it follows that $\Theta (1-\Theta_{01}) = \Theta (1-\Theta_0) + \Theta \Theta_{00}$ is trace-class.
Applying $(1-\Theta_{01})^{-1}$ to both sides (on the right), we conclude that $\Theta$ is trace-class.
The proof that $[H_1, P] \Phi (H_1)$ is trace-class is 
the same as in Theorem \ref{prop:stabilitysmall} (the smallness condition on $\mu$ was not used for this).
\end{proof}

Using Lemmas \ref{lemma:resolvent} and \ref{lemma:decayPhiH}, we see that for any $s \in \mathbb{R}$ and $p>0$,
$W \in \Op (S (\aver{\xi, \zeta - A_2 (x)}\aver{x, \zeta}^s \aver{y}^{-1-p}))$ satisfies the assumptions of Theorem \ref{prop:stabilitysmall} and 
$W \in \Op (S (\aver{\xi, \zeta - A_2 (x)}^{1-p} \aver{x}^{-p} \aver{y}^{-1-p}))$ satisfies the assumptions of Theorem \ref{prop:stabilitycompact}.
Note that the assumption that $\Phi_0 (H_0) W$ is trace-class forces the symbol of $W$ to decay in $y$, as the symbol of $H_0$ is independent of this variable.
We now introduce a stability result that no longer requires this decay in $y$.
\begin{theorem}\label{prop:stabilitystrip}
Let $W = W(x,y)$ be a symmetric point-wise multiplication operator in $S(1)$ that is compactly supported in $x$.
Then $\sigma_I (H+W) = \sigma_I (H)$.
\end{theorem}
\begin{proof}
Note that $H+W$ is self-adjoint with domain of definition $\mathcal{D} (H+W) = (i-(H+W))^{-1} \mathcal{H}$, where $\mathcal{H} = L^2 (\mathbb{R}^2) \otimes \mathbb{C}^2$ \cite{thaller2013dirac}.
Hence $\varphi ' (H+W)$ is well defined by the functional calculus.

Fix $\eps > 0$,
let $\chi \in \mathcal{C}^\infty_c (\mathbb{R})$ and define $\chi_\eps (y) := \chi (\eps y)$.
Then $\chi_\eps W$ satisfies the assumptions of Theorem \ref{prop:stabilitycompact} and thus $\sigma_I (H + \chi_\eps W) = \sigma_I (H)$. Since $W$ and $\chi_\eps W$ commute with $P$, we can write
\begin{align*}
    \sigma_I (H_W) - \sigma_I (H_\eps) = \Tr i [H,P] (\varphi ' (H_W) - \varphi ' (H_\eps)),
\end{align*}
where $H_W := H+W$ and $H_\eps := H+\chi_\eps W$.
The Helffer-Sj\"ostrand formula implies that
\begin{align*}
    \varphi ' (H_W) - \varphi ' (H_\eps) = \frac{1}{\pi} \int_{\mathbb{C}} \bar{\partial} \tilde{\varphi '} (z) (z-H_W)^{-1} (1-\chi_\eps)W (z-H_\eps)^{-1} d^2 z.
\end{align*}
For $\eps$ sufficiently small so that $[H,P] (1-\chi_\eps) = 0$, we thus have
\begin{align*}
    [H,P] (\varphi ' (H_W) &- \varphi ' (H_\eps)) =\\
    &\frac{1}{\pi} [H,P]\int_{\mathbb{C}} \bar{\partial} \tilde{\varphi '} (z) (z-H_W)^{-1} [H_W,\chi_\eps] (z-H_W)^{-1} W (z-H_\eps)^{-1} d^2 z.
\end{align*}
Since $[H_W, \chi_\eps] = [H,\chi_\eps] = -i \eps (\chi ')_\eps \sigma_2$ with $(\chi ')_\eps (y) = \chi ' (\eps y)$, it is clear that $$[H,P] (\varphi ' (H_W) - \varphi ' (H_\eps)) \in \Op (\eps S (1)).$$
Since $\chi_\eps W \in \mathcal{C}^\infty_c (\mathbb{R}^2)$, the proof of Lemma \ref{lemma:decayPhiH} 
implies that $\varphi ' (H_\eps) \in \Op (S (\aver{x,\xi,\zeta})^{-\infty})$ uniformly in $\eps$.
By interpolation, we conclude that $\sigma_I (H_W) - \sigma_I (H_\eps) \rightarrow 0$ as $\eps \rightarrow 0$, as desired.
\end{proof}

We have provided three stability results, each of which has a different assumption on the perturbation $W$. We now compare these assumptions using three illustrative examples.

\textbf{Example 1.}
Suppose $W = W(y) = \aver{y}^{-1-\delta}$ for some $\delta > 0$. Then Lemma \ref{lemma:decayPhiH} and the composition calculus imply that for any $p>0$, $\Phi_0 (H_0)W \in \Op (S (\aver{y}^{-1-\delta}\aver{x,\xi, \zeta}^{-p}))$.
It follows that the symbol of $\Phi_0 (H_0) W$ is integrable, meaning that $\Phi_0 (H_0) W$ is trace-class. Since $W$ is bounded, we see that $W$ satisfies the assumptions of Theorem \ref{prop:stabilitysmall}. This means $\sigma_I (H+\mu W) = \sigma_I (H)$ for all $\mu > 0$ sufficiently small. But note that $(i + H_0)^{-1} W$ is not compact, as its symbol does not decay in $\aver{x,\zeta}$. Thus Theorem \ref{prop:stabilitycompact} does not apply.
It is also clear that Theorem \ref{prop:stabilitystrip} does not apply since $W$ is independent of $x$.

\textbf{Example 2.}
Suppose $W = W(x,y) = a \aver{x,y}^{-2-\delta}$ for some positive constants $a$ and $\delta$.
As above it follows that $\Phi_0 (H_0) W$ is trace-class.
Moreover, since $W$ decays in both $x$ and $y$, the composition calculus implies that $(i \pm H_0)^{-1} W \in\Op (S (\aver{x,y,\xi,\zeta}^{-2}))$, meaning that $(i \pm H_0)^{-1} W$ is compact. Thus $W$ satisfies the assumptions of Theorem \ref{prop:stabilitycompact}.
Although $W$ also satisfies the assumptions of Theorem \ref{prop:stabilitysmall}, the latter cannot be used to prove $\sigma_I (H+W) = \sigma_I (H)$ when $a$ is sufficiently large.
Again it is clear that Theorem \ref{prop:stabilitystrip} does not apply, as $W$ is not compactly supported in $x$.

\textbf{Example 3.}
Suppose $W(x) = W \in \mathcal{C}^\infty_c (\mathbb{R})$, so that Theorem \ref{prop:stabilitystrip} applies. Then the symbol of $\Phi_0 (H_0) W$ does not decay in $y$, meaning that $\Phi_0 (H_0) W$ is not trace-class. Thus $W$ does not satisfy the assumptions of Theorem \ref{prop:stabilitysmall} or \ref{prop:stabilitycompact}.

\section*{Acknowledgment.} This research was partially supported by the U.S. National Science Foundation under Grants DMS-1908736 and EFMA-1641100.

\appendix
\section{Pseudo-differential calculus and the Helffer-Sj\"ostrand formula}\label{sec:PsiDO}
We briefly define the notation used in Section \ref{sec:stability} regarding pseudo-differential operators ($\Psi$DOs). For a more detailed exposition, see \cite[section 2]{QB} and references therein.

Given a 
symbol $a(x,\xi) =a\in \mathcal{S}' (\mathbb{R}^d \times \mathbb{R}^d) \otimes \mathbb{M}_n$, we define
the Weyl quantization of $a$ as the operator 
\begin{align}\label{eq:weylquant}
    \Op (a) \psi (x) :=
    \frac{1}{(2\pi)^d} \int_{\mathbb{R}^{2d}}
    e^{i(x-y)\cdot \xi}
    a(\frac{x+y}{2}, \xi) \psi (y) dy d\xi,
    \qquad
    \psi \in \mathcal{S} (\mathbb{R}^d) \otimes \mathbb{C}^n.
\end{align}
Here, $\mathcal{S}$ denotes the Schwartz space and $\mathbb{M}_n$ the space of Hermitian $n \times n$ matrices.

A function $u: \mathbb{R}^{2d} \rightarrow [0,\infty)$ is called an \emph{order function} if there exist constants $C_0 > 0$, $N_0 > 0$ such that
$u(X) \le C_0 \aver{X-Y}^{N_0} u(Y)$ for all $X,Y \in \mathbb{R}^{2d}$. Here we use the notation $\aver{X} := \sqrt{1 + |X|^2}$.
Note that if $u_1$ and $u_2$ are order functions, then so is $u_1 u_2$.

We say that $a \in S(u)$ for $u$ an order function if
for every $\alpha \in \mathbb{N}^{2d}$, there exists $C_\alpha > 0$ such that $|\partial^\alpha a (X)| \le C_\alpha u(X)$
for all $X \in \mathbb{R}^{2d}$.
We write $S(u^{-\infty})$ to denote the intersection over $s \in \mathbb{N}$ of $S(u^{-s})$.

By \cite[Chapter 7]{DS}, we know that if $a \in S(u_1)$ and $b \in S(u_2)$, then $\Op (c) := \Op (a) \Op (b)$ is a $\Psi$DO, with 
\begin{align*}
    c (x,\xi) =
    (a \sharp b) (x,\xi) :=
    \Big( e^{\frac{i}{2} (\partial_x \cdot \partial_\zeta - \partial_y \cdot \partial_\xi)}
    a(x,\xi) b(y,\zeta) \Big)
    \Big \vert_{y=x,\zeta=\xi}
\end{align*}
and $c \in S(u_1 u_2)$.
We write $A \in \Op (S (u))$ to mean that $A = \Op (a)$ for some $a\in S(u)$.

Given $\phi \in \mathcal{C}^\infty_0 (\mathbb{R})$, there exists an almost analytic extension $\tilde{\phi} \in \mathcal{C}^\infty_0 (\mathbb{C})$ that satisfies
\begin{align} \label{aae}
    |\bar{\partial} \tilde{\phi}| \le C_N |\Im z|^N, &\quad 
    N \in \{0,1,2,\dots\}; \qquad
    \tilde{\phi} (\lambda) = \phi(\lambda), \quad \lambda \in \mathbb{R}.
\end{align}
We now recall \cite[Theorem 8.1]{DS}.
If $H$ is a self-adjoint operator on a Hilbert space, then
\begin{align} \label{HSformula}
    \phi (H) =
    -\frac{1}{\pi} \int \bar{\partial} \tilde{\phi} (z) (z-H)^{-1} d^2 z,
\end{align}
where $\bar{\partial} := \frac{1}{2} \partial_{\Re z} + \frac{i}{2} \partial_{\Im z}$ and
$d^2 z$ is the Lebesgue measure on $\mathbb{C}$.
(\ref{HSformula}) is known as the Helffer-Sj\"ostrand formula.

\bibliographystyle{siam}
\bibliography{refs} 

\end{document}